\documentclass[a4paper,notitlepage,twocolumn,11pt,accepted=2021-08-25]{quantumarticle}
\pdfoutput=1
 
\usepackage[utf8]{inputenc}
\usepackage[english]{babel}
\usepackage[T1]{fontenc}
\usepackage{amsmath}
\usepackage{hyperref}
\usepackage[numbers, sort&compress]{natbib}

\usepackage{tikz}
      
\usepackage{url}          
\usepackage{booktabs}        
\usepackage{amsfonts}        
\usepackage{nicefrac}        
\usepackage{microtype}       
\usepackage{bm}
\usepackage{braket}
\usepackage{graphicx}
\usepackage{amsthm}
\usepackage{amssymb}
\usepackage{lipsum,mwe}

\usepackage{mathtools}
\usepackage{bbm}
\usepackage[explicit]{titlesec}
\usepackage{nccmath}
\usepackage{xcolor}

\usepackage{tkz-tab}
\usepackage{subcaption}
\usepackage{wrapfig}
\usepackage{enumitem}
\usepackage{caption}
\usepackage[linesnumbered,ruled,vlined]{algorithm2e}

\usepackage{appendix}

\usepackage{adjustbox}
\usepackage[normalem]{ulem}
\newtheorem{theorem}{Theorem}
\newtheorem{lemma}{Lemma}

\DeclareMathOperator{\Tr}{Tr}

\DeclareMathOperator{\Ber}{Ber}

\DeclareMathOperator{\diag}{diag}

\DeclareMathOperator{\Median}{median}
\DeclareMathOperator{\clip}{c}
\DeclareMathOperator{\flip}{f}
\DeclareMathOperator{\shift}{s}

\DeclareMathOperator{\Fnorm}{F}
\DeclareMathOperator{\n}{n}

\DeclareMathOperator{\Ker}{K}
\DeclareMathOperator{\Qer}{W}
\DeclareMathOperator{\Var}{Var}

\begin{document}
	
	\title{Towards understanding the power of quantum kernels in the NISQ era}
	\author{Xinbiao Wang}
	\thanks{Equal contribution; This work was done when he was a research intern at JD Explore Academy}
	\affiliation{Institute of Artificial Intelligence, School of Computer Science, Wuhan University}
	\affiliation{JD Explore Academy}
	\orcid{0000-0002-9898-820X}
	\author{Yuxuan Du}
	\thanks{Equal contribution; Corresponding author}
	\email{duyuxuan123@gmail.com}
	\affiliation{JD Explore Academy}
	\author{Yong Luo}
	\thanks{Corresponding author}
	\email{yong.luo@whu.edu.cn}
	\affiliation{Institute of Artificial Intelligence, School of Computer Science, Wuhan University}
	\author{Dacheng Tao}
	\thanks{Corresponding author}
	\email{dacheng.tao@sydney.edu.au}
	\affiliation{JD Explore Academy}
	
	\maketitle
	
	\begin{abstract}   
		A key problem in the field of quantum computing is understanding whether quantum machine learning (QML) models implemented on  noisy intermediate-scale quantum (NISQ) machines can achieve quantum advantages. Recently, Huang et al. [Nat Commun 12, 2631] partially answered this question by the lens of quantum kernel learning. Namely, they exhibited that quantum kernels can learn specific datasets with lower generalization error over the optimal classical kernel methods. However, most of their results are established on the ideal setting and ignore the caveats of near-term quantum machines. To this end, a crucial open question is: \textit{does the power of quantum kernels still hold under the NISQ setting?} In this study, we fill this knowledge gap by exploiting the power of quantum kernels when the quantum system noise and sample error are considered. Concretely, we first prove that the advantage of quantum kernels vanishes for large size of datasets, few number of measurements, and large system noise. With the aim of preserving the superiority of quantum kernels in the NISQ era, we further devise an effective method via indefinite kernel learning. Numerical simulations accord with our theoretical results. Our work provides theoretical guidance of exploring advanced quantum kernels to attain quantum advantages on NISQ devices. 
		
	\end{abstract}  
	

\section{Introduction}	
Kernel methods provide powerful framework to perform nonlinear and nonparametric learning, attributed to their universal property and interpretability \cite{boser1992training,hofmann2008kernel,shawe2004kernel}. During the past decades, kernel methods have been broadly applied to accomplish image processing, translation, and data mining tasks \cite{takeda2007kernel,sewell2013translation,anand1997designing}. As shown in Figure~\ref{fig:scheme}, a general rule of kernel methods is embedding the given input $\bm{x}^{(i)}\in\mathbb{R}^d$ into a high-dimensional feature space, i.e., $\phi(\cdot):\mathbb{R}^d\rightarrow \mathbb{R}^q$ with $q \gg d$, which allows that different classes of data points can be readily separable. Note that explicitly manipulating of $\phi(\bm{x}^{(i)})$ becomes computationally expensive for large $q$. To permit efficiency, kernel methods construct a kernel matrix $\Ker \in \mathbb{R}^{n\times n}$ to effectively accomplish the learning tasks in the feature space with $n$ being the size of training examples. Specifically, the elements of $\Ker$ represent the inner product of feature maps with $\Ker_{ij}=\Ker_{ji}=\langle \phi(\bm{x}^{(i)}), \phi(\bm{x}^{(j)}) \rangle$ for $\forall i,j\in[n]$, where such an inner product can be evaluated by a positive definite function  $\kappa(\bm{x}^{(i)},\bm{x}^{(j)})$  in $O(d)$ runtime. The performance of kernel methods heavily depends on the utilized embedding function $\phi(\cdot)$, or equivalently the function $\kappa(\cdot, \cdot)$ \cite{genton2001classes,bishop2006pattern}. To this end, various kernels such as the radial basis function kernel, Gaussian kernel, circular kernel, and polynomial kernel have been proposed to tackle various tasks \cite{amari1999improving,khemchandani2007twin}. Moreover, a recent study showed that the evolution of neural networks during training can be described by the neural tangent kernel \cite{jacot2018neural}.

	Quantum machine learning (QML) aims to effectively solve certain learning tasks that are challenging for classical methods \cite{biamonte2017quantum,preskill2018quantum,wittek2014quantum}. Theoretical studies have demonstrated that many QML algorithms, e.g., quantum perceptron \cite{kapoor2016quantum},  quantum support vector machine \cite{li2019sublinear}, and quantum differentially private sparse learning  \cite{du2020quantum_dp}, outperform their classical counterparts in the measure of runtime complexity. Despite  runtime advantages, the required resources to implement these algorithms are expensive and even unaffordable for noisy intermediate-scale quantum (NISQ) machines \cite{preskill2018quantum,arute2019quantum}. Meanwhile, experimental studies have confirmed the feasibility of using near-term devices to accomplish various QML tasks such as classification \cite{havlivcek2019supervised} and image generation \cite{huang2020experimental,rudolph2020generation}, and drug design \cite{li2021quantum}. However, theoretical results to guarantee quantum advantages of these NISQ-based QML algorithms are lacking. Therefore, an open question in the field of QML is `What QML algorithms  can be executed on NISQ devices with evident advantages?'. 
	
	A possible solution towards the above question is quantum kernels \cite{schuld2021quantum}. As shown in Figure \ref{fig:scheme}, there is a close correspondence between classical kernels and quantum kernels: the feature map $\phi(\cdot)$ coincides with the preparation of quantum states via variational quantum circuits $U_E(\bm{x}^{(i)})\in\mathbb{C}^{2^N \times 2^N}$ \cite{benedetti2019parameterized,du2018expressive,cerezo2021variational_VQA}, i.e., $\ket{\varphi(\bm{x}^{(i)})}=U_E(\bm{x}^{(i)})\ket{\bm{0}}$, which map the input data into high-dimensional Hilbert spaces described by $N$ qubits but can not be effectively accessible; the  result of kernel function $\kappa(\cdot, \cdot)$ coincides with applying measurements on the prepared quantum states, i.e., $|\braket{\varphi(\bm{x}^{(j)})|\varphi(\bm{x}^{(i)})}|^2$,  which enables the efficient collection of information from feature space. Due to the flexibility of the variational quantum circuits, quantum kernels have been experimentally implemented on different platforms such as superconducting, optical, and nuclear magnetic resonance quantum chips to resolve classification tasks \cite{havlivcek2019supervised,kusumoto2019experimental,bartkiewicz2020experimental,zhang2020prot}. As indicated by \cite{schuld2019quantum},  quantum kernels can achieve advantages when the prepared quantum states are classically intractable. Following the same routine, Huang et al. \cite{huang2021power} recently proved the predication advantages of quantum kernels. Namely, for appropriate  datasets, quantum kernels assures a lower generalization error bound than that of classical kernels \cite{vapnik1992principles}.  
	
	\begin{figure*}[htbp]
		\centering 
		\includegraphics[width=16cm]{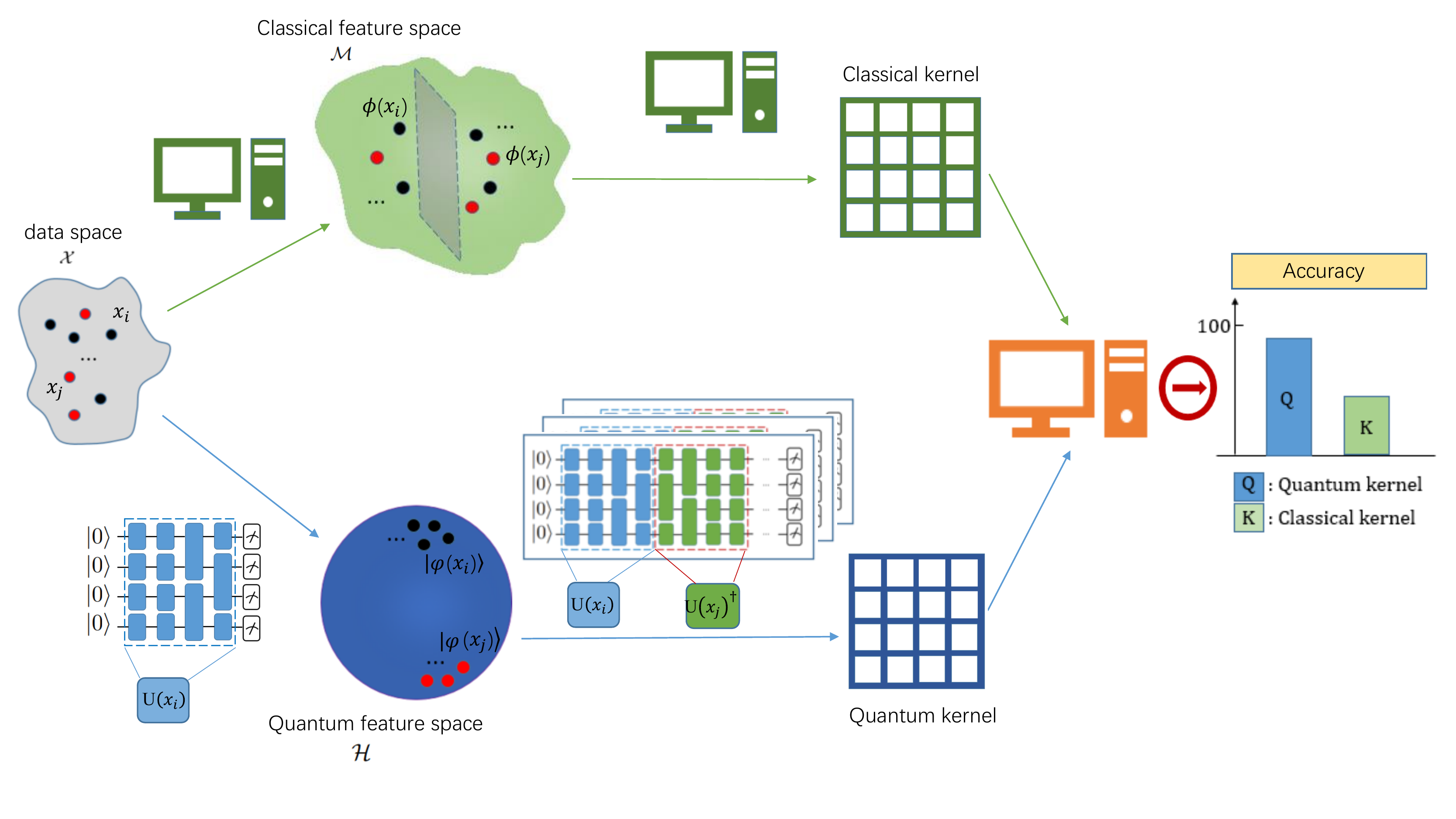}
		\caption{\small{\textbf{The paradigm of classical and quantum kernels.} Both of the classical and quantum kernels embed the data points from data space $\mathcal{X}$ into high-dimensional space, and then compute the kernel as the inner product of feature maps. The quantum kernel leverages variational quantum circuits to achieve this goal, as indicated by the blue color. In the ideal scenario, quantum kernels promise a better performance over classical kernels for certain datasets.    }}
		\label{fig:scheme}
	\end{figure*}
	
Despite the promising achievements, most of the theoretical results in \cite{huang2021power} are established on the ideal setting. In particular, they assumed that the number of measurements is infinite and the exploited quantum system is noiseless,   where both of them are impractical for NISQ devices. The quantum kernel returned by NISQ machines, affected by the system noise and a finite number of measurements, may be indefinite and therefore does not obey the results claimed in \cite{huang2021power}.  Driven by attractive merits comprised by quantum kernel methods  and the deficiencies of near-term quantum machines, a crucial question is: \textit{Does the power of quantum kernels still hold in the NISQ era?}  A positive affirmation of this question will not only contribute to a wide range of machine learning tasks to gain prediction advantages but can also establish the quantum deep learning theory.

A central theoretical contribution of this paper is exhibiting that a larger data size $n$, a higher system noise $p$, and a fewer number of measurements $m$  will make the generalization advantage of quantum kernels \textit{inconclusive}. This result indicates a negative conclusion of using quantum kernels implemented on NISQ devices to tackle large-scale learning tasks with evident advantages, which is contradicted with the claim of the study \cite{huang2021power} such that a larger data size $n$ promises a better generalization error. Moreover, we show that quantum system noise is a fatal factor that has the ability to collapse any superiority provided by quantum kernels. These observations are crucial guidance to help us design powerful quantum kernels to earn quantum advantages in the NISQ era.

Our second contribution is empirically demonstrating that under the NISQ setting, the advantages of quantum kernels may be preserved by suppressing its estimation error. Concretely, we adopt advanced spectral transformation techniques, which are developed in the indefinite kernel learning, to alleviate the negative effect induced by the system noise and the finite measurements. Numerical simulation results demonstrate that the performance of noisy quantum kernels can be improved by $14\%$. Our work opens up a promising avenue to combine classical indefinite kernel learning methods with quantum kernels to attain quantum advantages in the NISQ era.   
	
	\medskip
	\section{Quantum kernels in the NISQ scenario}  
	Before elaborating on our main results, let us first follow the study \cite{huang2021power} to formulate the quantum kernel learning tasks in the NISQ scenario. In particular,  suppose that both the training and test examples are sampled from the same domain $\mathcal{X}\times \mathcal{Y}$. The training dataset is denoted by  $\mathcal{D}=\{\bm{x}^{(i)}, y^{(i)}\}_{i=1}^n \subset \mathcal{X}\times \mathcal{Y}$, where $\bm{x}^{(i)}\in\mathbb{R}^d$ and $y^{(i)}\in\mathbb{R}$ refer to the $i$-th example with the feature dimension $d$ and the corresponding label, respectively. The prepared quantum state for the $i$-th example yields  $\ket{\varphi(\bm{x}^{(i)})}=U_E(\bm{x}^{(i)})\ket{0}^{\otimes N}$, where $U_E(\cdot)$ is the specified encoding quantum circuit and $N$ is the number of qubits. The relation between $\bm{x}^{(i)}$ and $y^{(i)}$ is $
	y^{(i)}=f(\bm{x}^{(i)}):=\Tr(OU(\bm{\theta}^*)\rho(\bm{x}^{(i)})U(\bm{\theta}^*)^{\dagger})$, where  $O$, $U(\bm{\theta}^*)$, and $\rho(\bm{x}^{(i)})$ represent the measurement operator, a specified quantum neural networks \cite{du2020learnability}, and the density operator of the encoded quantum data with $\rho(\bm{x}^{(i)})=\ket{\varphi(\bm{x}^{(i)})}\bra{\varphi(\bm{x}^{(i)})} \in \mathbb{C}^{2^N\times 2^N}$, respectively. The aim of quantum kernels learning is using the quantum kernel $\Qer\in \mathbb{R}^{n\times n}$, i.e.,
	\begin{equation}\label{eqn:Q-kernel}
		\Qer_{ij}=\Tr(\rho(\bm{x}^{(i)}) \rho(\bm{x}^{(j)})),~\forall i,j\in[n].
	\end{equation}
	to infer a hypothesis $h(\bm{x}^{(i)})= \langle \bm{\omega}^*, \varphi(\bm{x}^{(i)}) \rangle$ with a low  generalization error, where $\bm{\omega}^*$ refers to the optimal parameter with $\bm{\omega}^* = \arg \min_{\bm{\omega}} \lambda \braket{\bm{\omega}, \bm{\omega}} + \sum_{i=1}^{N}(\braket{\bm{\omega}, \varphi(\bm{x}^{(i)})} - y^{(i)})^2$  (see Appendix \ref{append:review-ideal} for details). The \textit{generalization error} of quantum kernels is quantified by 
	\begin{equation}
		\mathbb{E}_{\bm{x}, \Qer}(|h(\bm{x})-y|),
	\end{equation}
	where the randomness is taken over the dataset and quantum kernels methods.
	
	We note that the quantum kernel $\Qer$ in Eqn.~\eqref{eqn:Q-kernel} corresponds to the ideal setting. Nevertheless, NISQ machines are prone to having errors and only support finite number of measurements \cite{preskill2018quantum}. In the worst scenario, the system noise can be simulated by the quantum depolarization channel $\mathcal{N}_p$, i.e.,
	\begin{equation}\label{eq:dep}
		\mathcal{N}_p(\rho(\bm{x}^{(i)}))=(1-{p})\rho(\bm{x}^{(i)})+\frac{ {p} \mathbb{I}_{2^N}}{2^N},
	\end{equation}
	where $p$ refers to the depolarization rate and $\mathbb{I}_{2^N}$ is identity. Note that $p=1-(1-\widetilde{p})^{L_Q}$ depends on the quantum circuit depth $L_Q$ and the depolarization rate $\widetilde{p}$ in each layer (see Appendix \ref{append:thm1} for details).   To this end, the element $\Qer_{ij}$ in the quantum kernel $\Qer$ changes to $\widetilde{\Qer}_{ij}=\Tr(\mathcal{N}_p(\rho(\bm{x}^{(i)}) \rho(\bm{x}^{(j)})))$ for $\forall i,j\in [n]$. In addition, when the number of measurements applied to $\widetilde{\Qer}$ is $m$, the estimated element of quantum kernels yields 
	\begin{equation}\label{eqn:est-ele-ker}
		\widehat{\Qer}_{ij} =\frac{1}{m}\sum_{k=1}^m V_k,
	\end{equation}   
	where $V_k\sim \Ber(\widetilde{\Qer}_{ij})$  is the output of a quantum measurement and $X\sim \Ber(p)$ refers to the Bernoulli distribution with $\Pr(X=0)=p$ and $\Pr(X=1)=1-p$. The generalization error bound  under the noise setting described  above is summarized in the following theorem, whose proof is provided in Appendix~\ref{append:thm1}. 
	 
\begin{theorem}\label{thm:3.1}
Let the size of training dataset be $n$ and the number of measurements is $m$. Define $Y=[y_1, \cdots, y_n]^{\top}$ as the label vector and $c_{\Qer}=\|\Qer^{-1}\|_2$. Suppose the system noise is modeled by $\mathcal{N}_p$ in Eqn.~\eqref{eq:dep}.  With probability at least $1-\delta$, the noisy quantum kernel $\widehat{W}$    in Eqn.~\eqref{eqn:est-ele-ker} can be used to infer a hypothesis $h(\bm{x})$ with generalization error
			\begin{equation}\label{eqn:gene_err_noisy}
				\mathbb{E}_{\bm{x}, \widehat{\Qer}}\left|h(\bm{x})-y\right|  \leq  \tilde{O}\left(\sqrt{\frac{c_1}{n}}+ \sqrt{\frac{1}{c_2}\frac{n}{\sqrt{m}} } \right)
			\end{equation}
			where $c_1=Y^{\top}{\Qer}^{-1}Y$ and $c_2 = \max  ( c_{\Qer}^{-2}  (  (\frac{1}{2}\log  (\frac{4n^2}{\delta}  ) )^{\frac{1}{2}} + m^{\frac{1}{2}}p\left(1+\frac{1}{2^{N+1}}  ) \right)^{-1}- \frac{n}{\sqrt{m}}c_{\Qer}^{-1}, 0)$.

Notably, when the depolarization noise is considered, the generalization error of the noisy quantum kernel $\mathbb{E}_{\bm{x}, \widehat{\Qer}}|h(\bm{x})-y|$ will always have a term $n^{1/4}$.		

\end{theorem}

	\begin{figure*}[htbp]
		\centering
\includegraphics[width=0.98\textwidth]{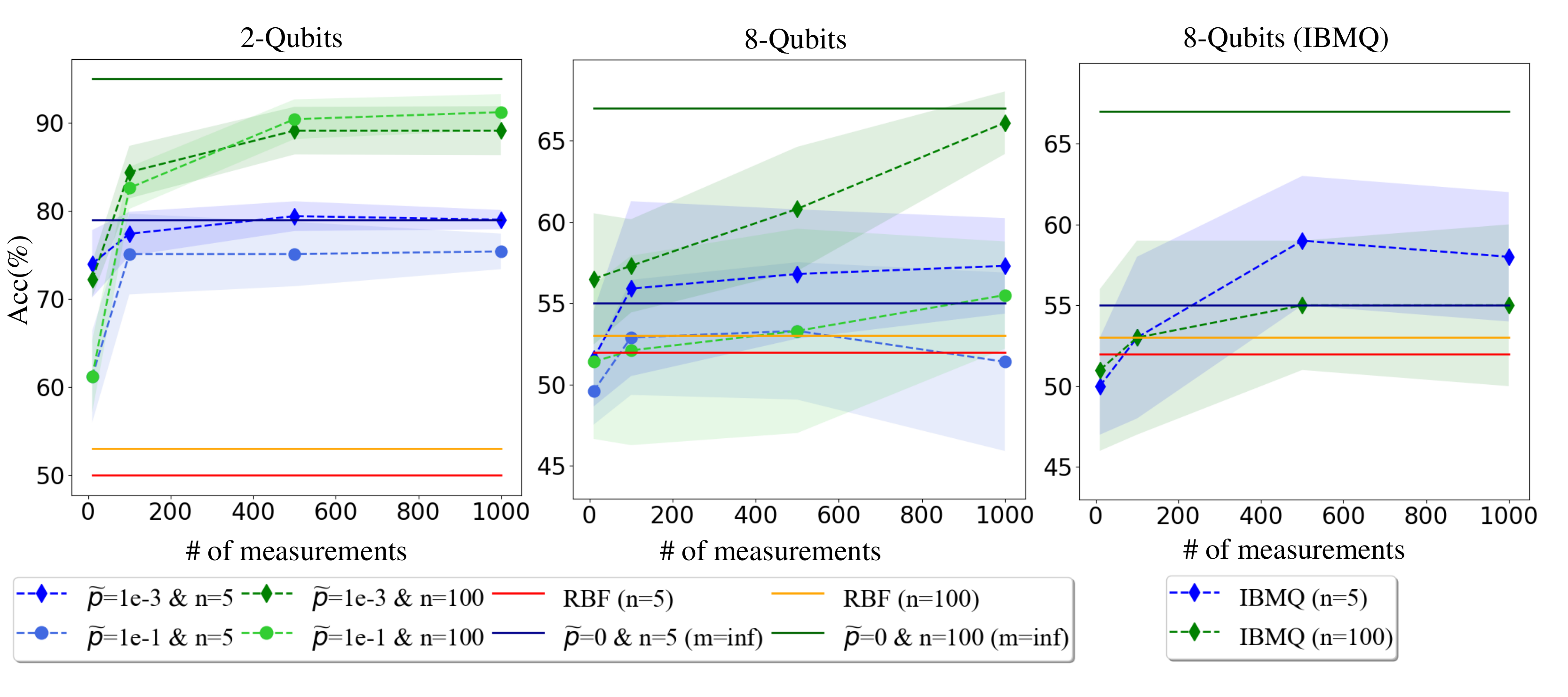}
		\caption{\small{\textbf{The performance of noisy quantum kernels on the engineered dataset.} The left and middle panels illustrate the  prediction accuracy (the higher means the better) of noisy quantum kernels under varied hyper-parameters settings when $N=2$ and $N=8$, respectively. The label `$\widetilde{p}=a \& n=b$' refers that the depolarization rate is $a$ and the number of training examples is $b$. Similarly, the label `RBF ($n=a$)' represents that $n=a$ training examples are used to train RBF kernel. The label `IBMQ ($n=a$)' represents that $n=a$ training examples are used to train noisy quantum kernel implemented on the real quantum hardware IBMQ-Melbourne. The shaded region refers to the standard deviation over $10$ independent runs.  
			 }}
		\label{fig:2qubit}
	\end{figure*}

 The results achieved in Theorem \ref{thm:3.1} indicate that the generalization error bound in Theorem \ref{thm:3.1} is nearly saturated in the NISQ setting, where the lower and upper bounds are separated by a factor $n^{1/4}$. In other words, the generalization error bound of noisy kernels must contain a term that is proportional to $n$. Such an observation implies that the generalization bound derived in \cite{huang2021power} fails to explain the generalization ability of noisy quantum kernels, since their result shows an upper bound of $O(\sqrt{c_1/n})$ in the noiseless setting and claims that the error continuously vanishes as $n$ enlarged. To further elucidate the separation between noisy and ideal quantum kernels, we first conduct numerical simulations to exhibit how the generalization error varies with different size of training dataset n under ideal and noise scenarios, respectively. Simulation results indicate that the prediction accuracy for the noisy kernel begins to decline when n exceeds a certain threshold, which implies the correctness of Eqn.~\eqref{eqn:gene_err_noisy}. We further theoretically derive that the achieved upper bound in Eqn.~\eqref{eqn:gene_err_noisy} is nearly saturated, which indicates that the generalization error bound of noisy kernels must contain a term that is proportional to n. Refer to Appendix D for more details.

Besides that, Theorem \ref{thm:3.1} provides the following two insights.
		\begin{itemize}
			\item The performance of quantum kernels in the NISQ era heavily depends on the number of measurement $m$. When $m=O(n^3)$ and $p$ is small, the generalization error of noisy quantum kernels is competitive with the ideal case. When $m < n$, the advantage of quantum kernels entirely vanishes. 
			\item The term $c_2$ in Eqn.~\eqref{eqn:gene_err_noisy} indicates the negative role of the system noise, i.e.,  even though $m$ is set as sufficiently large,  the generalization error bound can still be very large induced by $p$. Moreover, the generalization error bound in Eqn.~\eqref{eqn:gene_err_noisy} would be infinite once $ {p}> 1/(n c_{\Qer} (1+{1}/{2^{N+1}})) $.   
		\end{itemize}
		The above insights can be employed as guidance to design powerful quantum kernels in the NISQ era. In particular, the size of the training dataset $n$ should be controlled to be small. A possible solution is constructing a coreset \cite{bachem2017practical}. In addition, the minimum number of measurements $m$ is required to be set as $n^3$ to peruse potential quantum advantages. Last, since $p$ is determined by the circuit depth $L_Q$ and the depolarization rate $\widetilde{p}$ as shown in Eqn.~\eqref{eq:dep}, a good quantum kernel prefers a shallow circuit depth in $U_E$ and advanced error mitigation techniques \cite{temme2017error,endo2018practical,li2017efficient,du2020quantumQAS,strikis2020learning,czarnik2020error}.

\textbf{Remark.} 
The results in Theorem \ref{thm:3.1} is established on the depolarization noise, the achieved results can be easily extend to more general noisy channels. 
 
We next conduct extensive numerical simulations to validate the correctness of Theorem \ref{thm:3.1}. In particular, following the same routine as the study \cite{huang2021power} does, we adopt the Fashion-MNIST to exhibit how noise affects the superiority of quantum kernels. The data preprocessing stage contains four steps. First, we clean the dataset and only reserve images with labels `0' and `3', which correspond to the `cloth' and `dress', respectively. In other words, our simulation focuses on the binary classification task. Second, we sample $n$ ($n_{Te}$) examples from the filtered data to construct the training dataset $\mathcal{D}$ (test dataset $\mathcal{D}_{Te}$). Third, a feature reduction technique, i.e., principle component analysis (PCA) \cite{jolliffe1986principal}, is exploited to project each example in $\mathcal{D}\cup \mathcal{D}_{Te}$ (an image with feature dimension $28\times 28$) into a low-dimensional feature vector $\bm{x}_i \in \mathbb{R}^{N}$, where $N$ refers to the number of qubits. Last, we reassign the data label in $\mathcal{D}\cup \mathcal{D}_{Te}$ to saturate the geometric difference between quantum kernels and classical kernels following the method proposed by \cite{huang2021power}.  

Once the data preprocessing is completed, we apply quantum kernel methods to learn these modified datasets under both the noiseless and noisy settings. Furthermore, to understand whether noisy quantum kernels hold any superiority, we introduce an advanced classical kernel, i.e., the radial basis function (RBF) kernel, as a reference to learn these modified datasets. Note that RBF kernel is optimized by tuning hyper-parameter and regularization parameter to make full use their power (see Appendix~\ref{append:subsec:hyper-para} for details). The hyper-parameter setting used in the numerical simulations is as follows. The size of the training dataset $\mathcal{D}$ is set as $n \in \{5,  100 \}$. The size of the test dataset $\mathcal{D}_{Te}$ is 100. The number of measurement shots is set as $m\in \{10, 100, 500, 100\}$. The number of qubits, or equivalently the dimension of the projected feature $\bm{x}_i$, is set as $N\in\{2,8\}$. Notably, we adopt two types of noise models to quantify performance of noisy quantum kernels, which include the depolarization channel $\mathcal{N}_p$ in Eqn.~\eqref{eq:dep} with $\widetilde{p}\in \{0.001, 0.1\}$ and the noise model extracted from the real quantum hardware IBMQ-Melbourne.  Since the optimization of indefinite kernels is intractable, we adopt the nearest projection method~\cite{higham1988computing} to project the indefinite kernel onto positive definite matrix space.
Please refer to Appendix \ref{append:num-sim} for more simulation details.

The simulation results are illustrated in Figure \ref{fig:2qubit}.  Specifically,  the left panel exhibits the simulation results under the depolarization noise with $N=2$.  When $n=100$, the test accuracy is continuously approaching to the baseline $95\%$ with respect to the increased number of measurements $m$. For instance, when $m=10$ and $m=1000$, the test accuracy of the noisy quantum kernel with $\widetilde{p}=0.05$ is $61.2\%$ and $91.2\%$, respectively. Similar observations can be summarized for the case of $n=5$. Moreover, the performance for the case of $n=5$ is competitive and even better than that of $n=100$ when the number of measurements is small. Notably, for both the noisy and ideal quantum kernels with varied settings, their performance is superior to that of RBF kernels.  The middle panel shows the simulation results under the depolarization noise with $N=8$. Concretely, an increased number of measurements $m$ and a decreased rate of noise $\widetilde{p}$ enable the improved performance of noisy quantum kernels. Meanwhile, the noisy quantum kernels always outperform RBF kernels once $m>500$. All above observations echo with Theorem \ref{thm:3.1}. The right panel depicts the performance of noisy quantum kernels under the specific noise model extracted from the real quantum hardware IBMQ-Melbourne. The simulation results suggest that the rule implied in Theorem \ref{thm:3.1} can be employed as guidance to describe the behavior of noisy quantum kernels in realistic settings.

	\section{Enhance performance of noisy quantum kernels } 
Let us revisit the consequence of noisy quantum kernels as indicated in Theorem \ref{thm:3.1}. Specifically, to preserve the superiority of quantum kernels carried out on NISQ machines, we can either slim the size of training datasets or suppress the effects of noise. The solution to address the first issue is apparent, i.e., the construction of a coreset \cite{bachem2017practical}, while the strategy to emphasize the second issue remains obscure.  In the following, we investigate how to mitigate the negative influence of quantum noise to further enhance performance of noisy quantum kernels.   

Recall that the term $\sqrt{{n}/(c_2\sqrt{m})}$ in Eqn.~\eqref{eqn:gene_err_noisy} describes how the estimation error, i.e., $  \|\widehat{\Qer}^{-1} - {\Qer}^{-1}   \|_{\Fnorm}$, induced by the quantum noise influences the performance of noisy quantum kernels  (see proof of Theorem \ref{thm:3.1}). In other words, an effective strategy that has the ability to shrink the distance between the noiseless and noisy quantum kernels contributes to a better generalization error. Mathematically, if there exists a matrix $\widehat{\Qer}_{\diamond}$ satisfying 
\begin{equation}\label{eqn:Q-ker-psd}
 \left\|\widehat{\Qer}_{\diamond}^{-1} -  {\Qer}^{-1}   \right\|_{\Fnorm} \leq  \left\|\widehat{\Qer}^{-1} -  {\Qer}^{-1}  \right\|_{\Fnorm},
\end{equation} 
then this matrix can be used to infer a better hypothesis $h(\cdot)$ instead of $\widehat{\Qer}$. Seeking the target matrix $\widehat{\Qer}_{\diamond}$ has been extensively investigated in indefinite kernel learning \cite{ong2004learning, wu2005analysis, luss2009support, chen2009similarity}. The key insight to solve this problem is using certain spectral transformations to convert $\widehat{\Qer}$ into  $\widehat{\Qer}_{\diamond}$. To facilitate discussion, in the rest of the paper, we denote the spectral decomposition of $\Qer$ and $\widehat{\Qer}$ as  $\Qer = \sum_{i=1}^{n} \lambda_i \bm{v}_i \bm{v}_i^{\top}$ and $\widehat{\Qer} = \sum_{i=1}^n \widehat{\lambda}_i \bm{u}_i \bm{u}_i^{\top}$, where $\lambda_i$ and  $\widehat{\lambda}_i$ refer to the eigenvalues and $\bm{v}_i$ and $\bm{u}_i$ are the corresponding eigenvectors. Without loss of generality, we assume $\widehat{\lambda}_1 \ge \cdots \ge \widehat{\lambda}_r \ge 0 \ge \cdots \ge \widehat{\lambda}_n$.    

Here we explore three advanced spectral transformation techniques to acquire $\widehat{\Qer}_{\diamond}$ and prove their theoretical guarantees. In particular, the first  approach is clipping all negative eigenvalues in $\widehat{\Qer}$ to be zero \cite{wu2005analysis}, i.e., $\widehat{\Qer}_{\clip} = \sum_{i=1}^r \widehat{\lambda}_i \bm{u}_i \bm{u}_i^{\top}$. The second approach is flipping the sign of negative eigenvalues \cite{graepel1999classification}, i.e., $
	\widehat{\Qer}_{\flip} = \sum_{i=1}^r \widehat{\lambda}_i  \bm{u}_i \bm{u}_i^{\top} - 
	\sum_{i=1}^r  \widehat{\lambda}_i  \bm{u}_i \bm{u}_i^{\top}$. The third approach is shifting all eigenvalues by a positive constant \cite{roth2003optimal}, i.e., 	$\widehat{\Qer}_{s} = \sum_{i=1}^n (\widehat{\lambda}_i - \widehat{\lambda}_n) \bm{u}_i \bm{u}_i^{\top}$
	respectively, where $\widehat{\lambda}_{\n}$ is the minimum non-positive eigenvalue of $\widehat{\Qer}$.  The following lemma exhibits that the modified quantum kernels   $\{	\widehat{\Qer}_{\clip}, 	\widehat{\Qer}_{\flip},	\widehat{\Qer}_{s}\}$ can achieve a lower estimation error over the noisy quantum kernel $\widehat{\Qer}$.
	
	\begin{figure} 
	\centering
\includegraphics[width=0.48\textwidth]{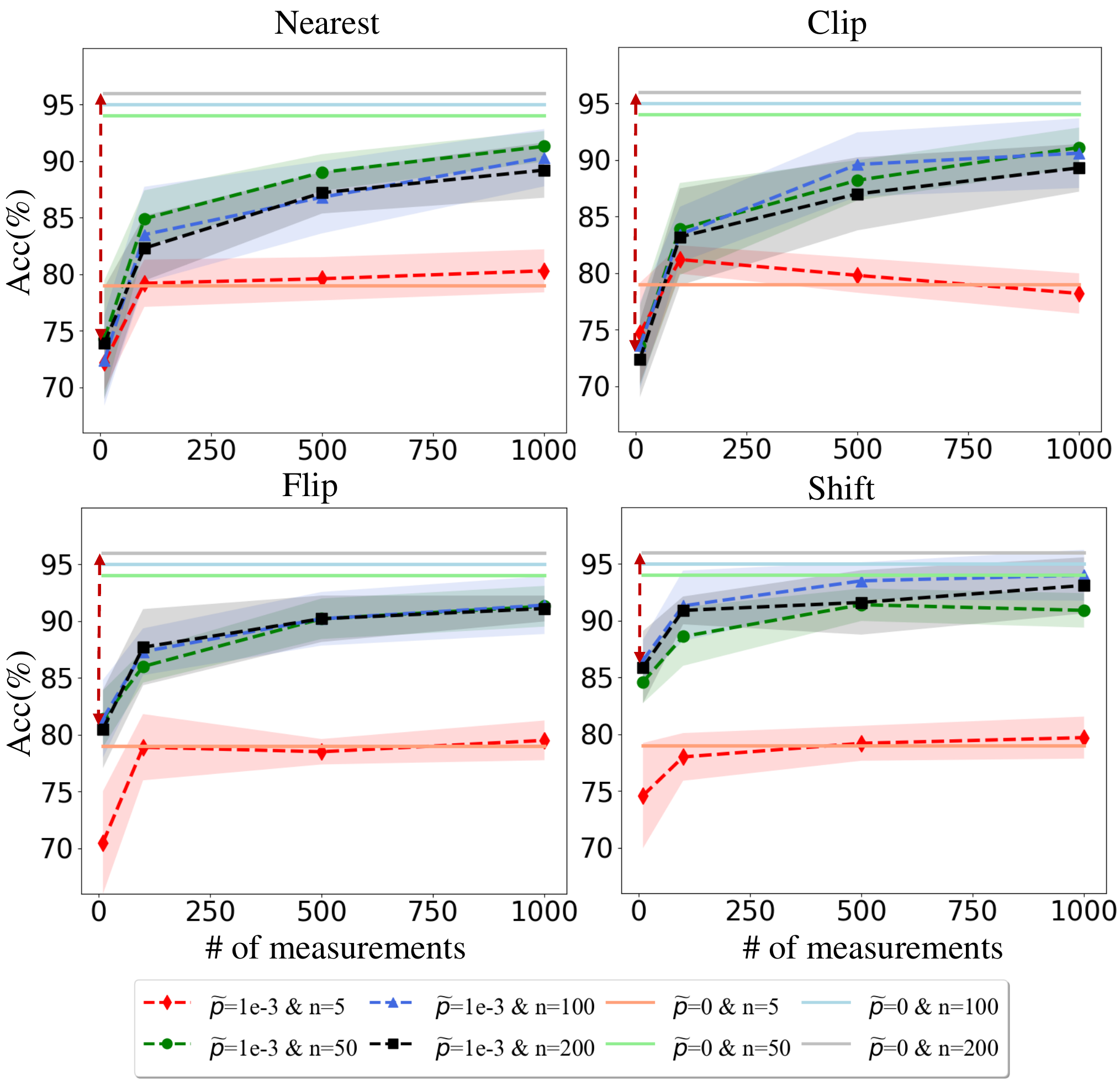}
	\caption{\small{\textbf{The comparison of noisy quantum kernels with different calibration methods.} The  simulation results of noisy quantum kernels calibrated by clipping, flipping, and shifting methods are shown in the upper left, upper right, lower left, and lower right, respectively. And the label `Original' refers to the case using nearest projection method. The meaning of different labels is same with those explained in Figure~\ref{fig:2qubit}. }}
	\label{fig:8qubit}
\end{figure}	
	
	\begin{lemma}\label{lem:proj-kernel}
	Let $\Qer$ and $\widehat{\Qer}$ be the ideal and  noisy quantum kernel in Eqns.(\eqref{eqn:Q-kernel}) and ~(\eqref{eqn:est-ele-ker}), respectively. Applying the  spectral transformation techniques to  $\widehat{\Qer}$, the obtained kernel $\widehat{\Qer}_{\diamond} \in \left \{ \widehat{ \Qer}_{\clip}, \widehat{ \Qer}_{\flip}, \widehat{ \Qer}_{\shift} \right \}$   yields 
		\begin{align}\label{eq:proj-kernel}
			\left \|\Qer - \widehat{\Qer}_{\diamond} \right \|_{\Fnorm} & \le \left \| \Qer - \widehat{ \Qer} \right \|_{\Fnorm},   
		\end{align}
		where $\left \| \cdot \right \|_{\Fnorm}$ refers to the Frobenius norm.
	\end{lemma}

 The proof of Lemma \ref{lem:proj-kernel} is given in Appendix \ref{append:proj-psd-kernel}.  Such a relation is ensured by the fact that $\widehat{\Qer}_{\diamond}$ can be viewed as the approximate orthogonal projection of $\widehat{\Qer}$ in the \textit{positive semi-definite} (PSD) space. The triangle inequality for the norm of the projection orthogonal immediately imply Eqn.~\eqref{eq:proj-kernel}. The result in Lemma \ref{lem:proj-kernel} suggests that spectral regularization may be a fundamental tool to improve generalization of quantum kernels in NISQ era.

We conduct the following numerical simulations to demonstrate that spectral transformation techniques contribute to enhance performance of noisy quantum kernels. Specifically, for all settings with $n\in\{5, 50, 100, 200\}$, the generalization performance of noisy quantum kernels in the original case is no better than $75\%$ when $m=10$, as shown in the upper left panel in Figure \ref{fig:8qubit}. By contrast, the training performance of noisy quantum kernels is dramatically improved when the spectral transformation techniques are adopted. Notably, the shifting method enables that the training performance of the noisy quantum kernels is competitive with the ideal quantum kernels, as shown in the lower right panel in Figure \ref{fig:8qubit}. The achieved simulation results provide a strong implication such that the superiority may be preserved by designing  advanced spectral transformation techniques. Please refer to Appendix \ref{append:num-sim} for the omitted details.

	\medskip
\section{Conclusion}
In this study, we investigate the generalization performance of quantum kernels under the NISQ setting. We theoretically exhibit that a large size of the training dataset, a small number of measurement shots, and a large amount of quantum system noise can destroy the superiority of quantum kernels. In addition, we demonstrate that the  generalization error bound in Theorem~\ref{thm:3.1} is nearly saturated. Our future work is tightening this upper bound. To improve performance of quantum kernels in the NISQ era, we further prove that effective spectral transformation techniques have the potential to maintain the advantage of quantum kernels in the NISQ era.  Besides the theoretical results, we empirically demonstrate that spectral transformation techniques have the capability of improving performance of noisy quantum kernels for both the depolarization noise and noise extracted from the real quantum-hardware (IBMQ-Melbourne). The achieved results in this study  fuel the exploration of quantum kernels assisted by other advanced calibration methods to accomplish practical tasks with advantages in the NISQ era.

\bibliographystyle{plainnat}

	\onecolumn\newpage
	\appendix
 	
The organization of the appendix is as follows. In Appendix \ref{append:notation}, we unify the notations used in the whole appendix. In appendix \ref{append:review-ideal}, we review the results of \cite{huang2021power} established in the noiseless setting. Then, in Appendix \ref{append:thm1} and Appendix~\ref{append:sec:saturation}, we exhibit the proof details of  Theorem \ref{thm:3.1} and compare the generalization error bound between ideal and noisy quantum kernels. Next, in Appendix \ref{append:proj-psd-kernel}, we prove Lemma \ref{lem:proj-kernel}, which shows that applying spectral transformation to the noisy quantum kernel contributes to reduce the kernel estimation error and thus narrow the generalization bound. Eventually, in Appendix \ref{append:num-sim}, we elaborate on the numerical simulation details.
	
	\section{The summary of notation}\label{append:notation}
	
	We unify the notations throughout the whole paper. The matrix is denoted by the capital letter e.g., $\Qer \in \mathbb{R}^{n \times n} $, The vector is denoted by the lower-case bold letter e.g., $\bm{x} \in \mathbb{R}^{d}$.  We denote $\braket{\cdot, \cdot}$ as the inner product in Hilbert space. The notation $[m]$ refers to the set $1,2, \cdots, m$. A random variable $X$ that follows Bernoulli distribution is denoted as $X \sim \Ber(p)$, i.e., $\Pr(X=1)=p$ and $\Pr(X=0)=1-p$. We denote $\| \cdot \|_2$ as the $\ell_2$-norm and $\| \cdot \|_{\Fnorm}$ as the Frobenius norm. 
	
	\section{The results of quantum kernels under the ideal setting}\label{append:review-ideal}
In this section, we recap the main results of the study \cite{huang2021power} to facilitate readers' understanding. Specifically, we first demonstrate the generalization of quantum kernels under the ideal setting. We next explain how to construct a specific dataset with quantum advantages in the measure of generalization. 
  	
\textit{\underline{Generalization error bound of ideal quantum kernels}.} The study \cite{huang2021power} focuses on the following setting. Namely, given $n$ data feature vectors $\{\bm{x}^{(i)}\}_{i=1}^n$ sampled from the domain $\mathcal{X}$, their corresponding labels yield 
\begin{equation}
	y^{(i)}=\Tr(U^{\dagger} O U \rho(\bm{x}^{(i)})) \in [-1,1],
\end{equation}
where $O$, $U(\bm{\theta}^*)$, and $\rho(\bm{x}^{(i)})=\ket{\varphi(\bm{x}^{(i)})}\bra{\varphi(\bm{x}^{(i)})}$ refer to the measurement operator, a specified quantum neural networks respectively. In other words, all labels $\{y^{(i)}\}_{i=1}^n$ can be well described by quantum circuits. 

Define a set of hypothesis functions as $\{h(\cdot)= \braket{\bm{\omega}, \varphi(\cdot)} | \bm{\omega} \in \Omega \}$, $\Omega$ is the parameters space and $\ket{\varphi(\cdot)}$ refers to the quantum feature map. The aim of quantum machine learning is seeking the optimal parameters $\bm{\omega}^*$ to enable a minimum empirical risk. This task can be achieved by optimizing the loss functions
\begin{equation}\label{eq:B2}
	\mathcal{L}(\bm{\omega}, \bm{x})=	\lambda \braket{\bm{\omega}, \bm{\omega}} + \sum_{i=1}^n ( \braket{\bm{\omega}, \varphi(\bm{x}^{(i)})} - \Tr(U^{\dagger} O U \rho(\bm{x}^{(i)})) )^2,
	\end{equation}
where $\lambda \ge 0$ is the hyper-parameter. Namely, the optimal parameters yield
\begin{equation}\label{eqn:opt-omega}
\bm{\omega}^* =	\arg\min_{\bm{\omega}\in \Theta}  \mathcal{L}(\bm{\omega}, \bm{x}).
\end{equation} 
The explicit form of $\bm{\omega}^*$ in Eqn.~\eqref{eqn:opt-omega} satisfy  
\begin{equation}
		\left\| \bm{\omega}^* \right\|_2^2 = Y^{\top} \Qer^{-1} Y,
\end{equation}
	where $Y=[y^{(1)}, ..., y^{(n)}]^{\top}$ refers to the vector of labels and $\Qer$ is the quantum kernels as defined in Eqn.~\eqref{eqn:Q-kernel}. Note that we always assume that $\Qer = \Qer + \lambda \mathbb{I}$ is non-singular with taking $\lambda \to 0$, which is a general assumption that is broadly employed in the study of quantum kernels.  
	
The key conclusion in \cite{huang2021power} is exhibiting that the generalization error bound of quantum kernels is quantified by $\left\| \bm{\omega}^* \right\|$.
	\begin{theorem}
		[Theorem 1, \cite{huang2021power}]\label{thm : huang}
		Define the given training dataset as $\{\bm{x}^{(i)}, y^{(i)}=\Tr(O^{U}\rho(\bm{x}^{(i)})) \}_{i=1}^n$ and the corresponding quantum kernel as $\Qer$ in Eqn.~\eqref{eqn:Q-kernel}. With probability at least $1-\delta$,
		quantum kernel methods can learn a hypothesis $h(\cdot)$ with the generalization error 		
		\begin{equation}\label{eqn:gene_error_ideal}
			\mathbb{E}_{(\bm{x}^{(i)}, y^{(i)})\sim \mathcal{X}\times\mathcal{Y} } \left|f(\bm{x}^{(i)})- y^{(i)}) \right| \leq  \tilde{O} \left(
			\sqrt{ \frac{ c_1}{n} }
			\right)
		\end{equation}
		where the randomness comes from  of the sampling training data from $\mathcal{X}$,  $c_1 \equiv\left\| \bm{\omega}^* \right\| =Y^{\top} \Qer^{-1} Y$ and the notation $\tilde{O}$ hides the logarithmic terms.
	\end{theorem}	
	
\textit{\underline{The construction of dataset with quantum advantages}}. Recall that the result of theorem \ref{thm : huang} indicates that the generalization error bound depends on the kernel matrix $\Qer$, the labels $Y$, and the size of the training dataset $n$. When the training examples $\{\bm{x}_i\}_{i=1}^n$ are fixed and the labels of data $Y$ can be modified, the study \cite{huang2021power} shows that the advantage of quantum kernels can be achieved by maximizing the geometric difference   
	\begin{equation}\label{eqn:Geo-diff}
	    \frac{ Y^{\top} \Ker^{-1}  Y}{Y^{\top} \Qer^{-1} Y},
	\end{equation}
	where $\Qer$ and $\Ker$ refer to quantum and classical kernels. Ensured by Theorem \ref{thm : huang}, a large geometric difference indicates that quantum kernel methods can achieve better generalization error than classical kernel models.

\section{Proof of Theorem \ref{thm:3.1}}\label{append:thm1}
Here we present the proof of Theorem \ref{thm:3.1}, especially for the derivation of Eqn.~\eqref{eqn:gene_err_noisy}. Note that we defer the proof that the generalization error bound of the noisy quantum kernel always contains an unavoidable term $n^{1/4}$ in Appendix \ref{append:sec:saturation}.  

Before moving on to elaborate Theorem \ref{thm:3.1}, we first simplify the  depolarization noise model applied to the quantum kernels as described in the main text. In particular, for an $L^Q$-layer quantum circuit, all noise channel $\mathcal{N}_{\widetilde{p}}$ separately applied to each quantum circuit depth can be compressed to the last layer and presented through a new depolarization channel $\mathcal{N}_{p}$. 
	\begin{lemma}
		[Lemma 6, \cite{du2020learnability}]\label{lem:3.1}\label{lem:dep-gene}
		Let $\mathcal{N}_{\widetilde{p}}$ be the depolarization channel applied to each unitary layer $U_l(\bm{\theta})$. There always exists a depolarization channel $\mathcal{N}_{p}$ with ${p}=1-(1-\widetilde{p})^{L_Q}$ that satisfies 
	\begin{equation}
\mathcal{N}_{\widetilde{p}}(U_{L_Q}(\bm{\theta}) \cdots U_2(\bm{\theta}) \mathcal{N}_p(U_1(\bm{\theta})\rho U_1(\bm{\theta})^{\dagger})U_2(\bm{\theta})^{\dagger}\cdots U_{L_Q}(\bm{\theta})^{\dagger} ) = \mathcal{N}_{p}( U(\bm{\theta}) \rho U(\bm{\theta})^{\dagger} ),
	\end{equation}
where $\rho$ is the input quantum state.
\end{lemma}
The proof of the above lemma is deferred to Appendix~\ref{append:subsec-lemm2}.


Besides Lemma \ref{lem:3.1}, the proof of Theorem 1 employs the following lemmas, where the corresponding proofs are deferred to Appendix \ref{append:subsec-rkks-proof} and Appendix \ref{append:subsec-lemm3}, respectively.

    \begin{lemma}\label{lem:rkks-bound}
        Define the given training set as $\{\bm{x}^{(i)}, y^{(i)}=\Tr(O^U \rho(\bm{x}^{i})) \}$ and the noisy kernel as $\widehat{W}$ in Eqn.~(\ref{eqn:est-ele-ker}). With probability at least $1-\delta$, the noisy quantum kernel can learn a hypothesis ${h}(\cdot)$ with the generalization error
        \begin{equation}
            \mathbb{E}_{(\bm{x}^{(i)}, y^{(i)})\sim \mathcal{X}\times\mathcal{Y} } \left|{h}(\bm{x}^{(i)})- y^{(i)}) \right| \leq  \widetilde{O} \left(
			\sqrt{ \frac{ \widehat{c} }{n} }
			\right)
        \end{equation}
        where the randomness comes from the sampling training data and the noise in the NISQ scenario, $\widehat{c}= |Y^{\top} \widehat{\Qer}^{-1} Y|$ and the notation $\widetilde{O}$ hides the logarithm terms.
\end{lemma}
	
	\begin{lemma}\label{lem:3.4}
    Suppose the system noise is modeled by the depolarization channel $\mathcal{N}_p$ in Eqn.~\eqref{eq:dep}. Define the ideal quantum kernel as $\Qer$ with entry $\Qer_{ij} = \Tr(\rho(\bm{x}^{(i)}) \rho(\bm{x}^{(j)}))$ and the noisy quantum kernel as $\widehat{\Qer}$ with entry $\widehat{\Qer}_{ij} = \frac{1}{m}\sum_{k=1}^m V_k$, where $V_k \sim \Ber(\widetilde{\Qer}_{ij})$ and $\widetilde{\Qer}_{ij} = \Tr(\mathcal{N}_p(\rho(\bm{x}^{(i)})\rho(\bm{x}^{(j)}))$. With probability at least $1-\frac{\delta}{2}$, we have
    	\begin{align} \label{eq:inverse_norm}
             \left \| {\Qer}^{-1} - \widehat{\Qer}^{-1} \right \|_2  \le     \frac{1}{c_2}\frac{n}{\sqrt{m}},        
		\end{align}
		where the randomness is taken over $\widehat{\Qer}$,  $c_2 = \max \left( c_{\Qer}^{-2} \left( \left(\frac{1}{2}\log \left(\frac{4n^2}{\delta} \right)\right)^{\frac{1}{2}} + m^{\frac{1}{2}}p\left(1+\frac{1}{2^{N+1}} \right) \right)^{-1}- \frac{n}{\sqrt{m}}c_{\Qer}^{-1}, 0 \right)$,
		and $c_{\Qer}=\|\Qer^{-1}\|$ is the spectral norm of $\Qer^{-1} $.
	\end{lemma}
	
We are now ready to exhibit the proof of Theorem \ref{thm:3.1}. 
\begin{proof}[Proof of Theorem \ref{thm:3.1}]
Following the result of Lemma \ref{lem:rkks-bound}, with probability at least $1-\frac{\delta}{2}$, the generalization error bound of the noisy quantum kernel $\widehat{\Qer}$ yields 
\begin{equation}\label{eq:C3}
			\mathbb{E}_{\bm{x}}\left|f(\bm{x})-y \right| \le \widetilde{O} \left(\sqrt{\frac{\widehat{c}}{n}}  \right)
	\end{equation}
where the randomness is taken over the sampling of the training data and $\widehat{c}=|Y^{\top}\widehat{\Qer}^{-1}Y|$.  Note that the upper bound of the term $\widehat{c}$ is as follow, i.e.,
\begin{align}\label{eq:C7}
		{}	& \sqrt{\left| Y^{\top}\widehat{\Qer}^{-1}Y \right| } \\
		= & \sqrt{Y^{\top}\left|\widehat{\Qer}^{-1} -  {\Qer}^{-1} \right|Y+Y^{\top}{\Qer}^{-1}Y }\\
		\le & \sqrt{\left \|\widehat{\Qer}^{-1} -  {\Qer}^{-1} \right \|_2 Y^{\top}Y +Y^{\top}{\Qer}^{-1}Y} \\
		\le & \sqrt{\left \|\widehat{\Qer}^{-1} -  {\Qer}^{-1} \right \|_2 Y^{\top}Y} + \sqrt{Y^{\top}{\Qer}^{-1}Y},  
		\end{align}
where the last inequality utilizes the fact that $\sqrt{a+b} \le \sqrt{a}+\sqrt{b}$ for any given $a,b \ge 0$. In other words, the generalization error of the noisy quantum kernel $\widehat{\Qer}$ is upper bounded by 
\begin{equation}\label{eq:C3-1}
\mathbb{E}_{\bm{x}}\left|f(\bm{x})-y \right| \le \widetilde{O} \left(\sqrt{\frac{ \|\widehat{\Qer}^{-1} -  {\Qer}^{-1} \|_2 Y^{\top}Y}{n}} + \sqrt{\frac{Y^{\top}{\Qer}^{-1}Y}{n}}  \right).
\end{equation}		
		
In the following, we derive the upper bound of the term  $\|\widehat{\Qer}^{-1} -  {\Qer}^{-1} \|_2$ to obtain the   generalization error bound of the noisy quantum kernel $\widehat{\Qer}$. Specifically, by leveraging Lemma \ref{lem:3.4}, with probability at least $1-\frac{\delta}{2}$, we have
	\begin{equation}\label{eq:C8}
            \left \| {\Qer}^{-1} - \widehat{\Qer}^{-1} \right \|_2 \le             \frac{1}{c_2}\frac{n}{\sqrt{m}}, 
\end{equation}
where the randomness is over $\widehat{\Qer}$.
		
In conjunction 	Eqn.~\eqref{eq:C3-1}  with Eqn.~\eqref{eq:C8}, we achieve the generalization error bound of the noisy quantum kernel $\widehat{\Qer}$, i.e., with probability at least $1-\delta$, 
		\begin{align}
			\mathbb{E}_{\bm{x}, \widehat{\Qer}}\left|f(\bm{x})-y \right| \le \widetilde{O} \left( \sqrt{\frac{c_1}{n}}+ \sqrt{\frac{1}{c_2}\frac{n}{\sqrt{m}}} \right)
		\end{align}
	where the randomness is over $\widehat{\Qer}$ and the sampling of the training data, $c_1=Y^{\top}{\Qer}^{-1}Y$, and $c_2 = \max \left( c_{\Qer}^{-2} \left( \left(\frac{1}{2}\log \left(\frac{4n^2}{\delta} \right)\right)^{\frac{1}{2}} + m^{\frac{1}{2}}p\left(1+\frac{1}{2^{N+1}} \right) \right)^{-1}- \frac{n}{\sqrt{m}}c_{\Qer}^{-1}, 0 \right)$.

\end{proof}
	
	\subsection{Proof of Lemma \ref{lem:3.1}}\label{append:subsec-lemm2}
	
\begin{proof}[Proof of Lemma~\ref{lem:3.1}]
    Denote $\rho^{(k)}$ as $\rho^{(k)} = \prod_{l=1}^k U_{l}(\bm{\theta}) \rho U_{l}(\bm{\theta})^{\dagger}$. Applying $\mathcal{N}_{\widetilde{p}}$ to $\rho^{(1)}$ yields
    \begin{align}
        \mathcal{N}_{\widetilde{p}}(\rho^{(1)}) = (1-{\widetilde{p}})\rho^{(1)} + {\widetilde{p}} \frac{\mathbb{I}_D}{D},
    \end{align}
    where $D$ refers to the dimensions of Hilbert space interacted with $\mathcal{N}_p$.
    
    By induction, we can suppose that at the $k$-th step, the generated state satisfies
    \begin{align}
        \rho^{(k)} = (1-{\widetilde{p}})\rho^{(k)} + (1-(1-{\widetilde{p}})^k) \frac{\mathbb{I}_D}{D}
    \end{align}
    Then applying $U_{k+1}(\bm{\theta})$ followed by $\mathcal{N}_p$ yields 
    \begin{align}
        \rho^{(k+1)} = \mathcal{N}_{{\widetilde{p}}}(U_{k+1}(\bm{\theta})\rho^{(k+1)} U_{k+1}(\bm{\theta})^{\dagger}) = (1-{\widetilde{p}})\rho^{(k+1)} + (1-(1-{\widetilde{p}})^{k+1}) \frac{\mathbb{I}_D}{D}.
    \end{align}
    According to the formula of depolarization channel, an immediate observation is that the noisy QNN is equivalent to applying a single depolarization channel $\mathcal{N}_p$ at the last circuit depth, i.e.,
    \begin{align}
    \mathcal{N}_{p}(\rho) = (1-{\widetilde{p}})^{L_Q}\rho^{(L_Q)} + (1-(1-{\widetilde{p}})^{L_Q})\frac{\mathbb{I}}{D},
    \end{align}
    where $p = 1 - (1-{\widetilde{p}})^{L_Q}$.
\end{proof}

\subsection{Proof of Lemma~\ref{lem:rkks-bound}}\label{append:subsec-rkks-proof}
The proof of Lemma~\ref{lem:rkks-bound} mainly employs a basic theorem in statistics and learning theory as presented below.
    \begin{theorem}[Theorem 3.3, \cite{mohri2018foundations}] Let $\mathcal{G}$ be a family of function mappings from a set $\mathcal{X}$ to $[0, 1]$. Then for any $\delta>0$, with probability at least $1-\delta$ over the identical and independent draw of $n$ samples from $\mathcal{X}:\bm{x}^{(1)}, \cdots, \bm{x}^{(n)}$, we have for all $g \in \mathcal{G}$,
\begin{equation} \label{eq:found_bound}
    \mathbb{E}_{\bm{x}}[g(\bm{x})] \le \frac{1}{n}\sum\limits_{i=1}^n g(\bm{x}_i) + 2\mathbb{E}_{\sigma}\left[\sup_{g \in \mathcal{G}}\frac{1}{n}\sum\limits_{i=1}^{n}\sigma_{i}g(\bm{x}^{(i)})\right] + 3 \sqrt{\frac{log(2/\delta)}{2n}} ,
\end{equation}
where $\sigma_1, \cdots, \sigma_n$ are independent and uniform variables over $\pm 1$.   
\end{theorem}

A discussion about the above theorem in \cite{huang2021power} gives a modified version, i.e.,  supported by Talagrand's lemma, the generalization error $\mathbb{E}_{(\bm{x}^{(i)}, y^{(i)})\sim \mathcal{X}\times\mathcal{Y} } |{h}(\bm{x}^{(i)})- y^{(i)}) | - \frac{1}{n} \sum_{i=1}^n |{h}(\bm{x}^{(i)})- y^{(i)}) |$ is upper bounded by
    \begin{equation}\label{eq:upper_bound_tala}
    2 \mathbb{E}_{\sigma} \left[ 
    \sup_{\|\alpha\|_1 \le \lceil\|\widehat{\alpha}\|_1\rceil} \frac{1}{n} \sum_{i=1}^n \sigma_{i}h_{\alpha}(\bm{x}_i)
    \right]
    + 6 \sqrt{\frac{\log (4\lceil\|\widehat{\alpha}\|_1\rceil^2/\delta)}{2n}},
\end{equation}
with probability at least $1-\delta$, where $\lceil u \rceil = \inf\{v \in \mathbb{Z} ~|~ u < v\}$ and $h_{\alpha}(\cdot) = \sum_{j=1}^{n}\alpha_j \widehat{\kappa}(\cdot, \bm{x}^{(j)})$ with $\widehat{\kappa}(\cdot, \cdot)$ being the indefinite kernel function related to $\widehat{\Qer}$. According to the theory of indefinite kernel learning \cite{ong2004learning}, $\widehat{\kappa}$ can be uniquely decomposed into the difference of two positive definite kernels $\widehat{\kappa}_{+}$ and $\widehat{\kappa}_{-}$. The sum of $\widehat{\kappa}_{+}$ and $\widehat{\kappa}_{-}$ gives rise to a positive definite kernel $\widehat{\kappa}^*$ 
which possesses the smallest reproducing  kernel Hilbert space corresponding to the Krein space of $\widehat{\kappa}$ \cite{ong2004learning}. 
We are now ready to present the proof of Lemma~\ref{lem:rkks-bound}.

    \begin{proof}[Proof of Lemma~\ref{lem:rkks-bound}]
We consider $\mathcal{G}=\{h_{\alpha}(\cdot) = \sum_{j=1}^{n}\alpha_j \widehat{\kappa}(\cdot, \bm{x}^{(j)}) ~| ~\|\alpha \|_{1} \le \lceil\|\widehat{\alpha}\|_1\rceil \}$, where $\|\widehat{\alpha}\|_1^2= |Y^{\top} \widehat{\Qer}^{-1} Y|$. With training error being zero, we have
    \begin{align}
        & \mathbb{E}_{(\bm{x}^{(i)}, y^{(i)})\sim \mathcal{X}\times\mathcal{Y} } \left|{h}(\bm{x}^{(i)})- y^{(i)}) \right| \nonumber \\
        \leq & 2
        \mathbb{E}_{\sigma}\left[\sup_{\|\alpha\|_1 \leq {[\|\widehat{\alpha}\|_1]} }
        \frac{1}{n} \sum_{i=1}^n \sigma_{i} h_{\alpha}(\bm{x}_i) \right] + 6 \sqrt{\frac{\log \left(4\lceil\|\widehat{\alpha}\|_1\rceil^{2} / \delta\right)}{2 n}} \nonumber \\
        = & 2 \mathbb{E}_{\sigma}\left[\sup_{\|\alpha\|_1 \leq {\lceil\|\widehat{\alpha}\|_1\rceil} } \frac{1}{n} \sum_{i=1}^{n} \sum_{j=1}^{n} \sigma_{i} \alpha_{j} \left( \widehat{\kappa}_{+}(\bm{x}^{(i)}, \bm{x}^{(j)}) - \widehat{\kappa}_{-}(\bm{x}^{(i)}, \bm{x}^{(j)}) \right) \right]+6 \sqrt{\frac{\log \left(4\lceil\|\widehat{\alpha}\|_1\rceil^{2} / \delta\right)}{2 n}}   \nonumber   \\
        \le & 2 \mathbb{E}_{\sigma}\left[\sup_{\|\alpha\|_1 \leq {\lceil\|\widehat{\alpha}\|_1\rceil} } \frac{1}{n} \sum_{i=1}^{n} \sum_{j=1}^{n} \sigma_{i} \alpha_{j} \left( \widehat{\kappa}_{+}(\bm{x}^{(i)}, \bm{x}^{(j)}) + \widehat{\kappa}_{-}(\bm{x}^{(i)}, \bm{x}^{(j)}) \right) \right]+6 \sqrt{\frac{\log \left(4\lceil\|\widehat{\alpha}\|_1\rceil^{2} / \delta\right)}{2 n}}   \nonumber   \\
        = & 2 \mathbb{E}_{\sigma}\left[\sup_{\|\alpha\|_1 \leq {\lceil\|\widehat{\alpha}\|_1\rceil} } \frac{1}{n} \sum_{i=1}^{n} \sum_{j=1}^{n} \sigma_{i} \alpha_{j}  \widehat{\kappa}^{*}(\bm{x}^{(i)}, \bm{x}^{(j)}) \right]+6 \sqrt{\frac{\log \left(4\lceil\|\widehat{\alpha}\|_1\rceil^{2} / \delta\right)}{2 n}}   \nonumber   \\
        = & 2 \mathbb{E}_{\sigma}\left[\sup_{\|\omega\|_2 \leq {\lceil\|\widehat{\alpha}\|_1\rceil} } \frac{1}{n} \sum_{i=1}^{n}  \sigma_{i} \braket{\omega, \phi^*(\bm{x}_i)}  \right]+6 \sqrt{\frac{\log \left(4\lceil\|\widehat{\alpha}\|_1\rceil^{2} / \delta\right)}{2 n}}   \nonumber   \\
         \leq & 2 \mathbb{E}_{\sigma}\left[\sup_{\|\omega\|_2 \leq {\lceil\|\widehat{\alpha}\|_1\rceil} } \frac{1}{n}\| \omega\|_2
          \left\|\sum_{i=1}^{n}  \sigma_{i} \phi^*(\bm{x}_i)\right\|   \right]  +6 \sqrt{\frac{\log \left(4\lceil\|\widehat{\alpha}\|_1\rceil^{2} / \delta\right)}{2 n}}    \nonumber   \\
        \leq & 2 \frac{ {\lceil\|\widehat{\alpha}\|_1\rceil} }{n} \sqrt{\mathbb{E}_{\sigma} \sum_{i=1}^{n} \sum_{j=1}^{n} \sigma_{i} \sigma_{j} \widehat{\kappa}^*\left(\bm{x}^{(i)}, \bm{x}^{(j)}  \right)} +6 \sqrt{\frac{\log \left(4\lceil\|\widehat{\alpha}\|_1\rceil^{2} / \delta\right)}{2 n}}     \nonumber   \\
        \leq & 2 \sqrt{\frac{\lceil\|\widehat{\alpha}\|_1\rceil^2\Tr(\widehat{\Qer}^* ) }{n}}+6 \sqrt{\frac{\log (\lceil\|\widehat{\alpha}\|_1\rceil )}{n}}+6 \sqrt{\frac{\log (4 / \delta)}{2 n}}     \nonumber   \\
        = & \widetilde{O} \left( \sqrt{\frac{ \lceil\|\widehat{\alpha}\|_1\rceil^2 }{n}} \right),
\end{align}
where $\omega = \sum_{j=1}^n \alpha_j \phi^*(\bm{x}_j)$ satisfies $\|\omega\|_2^2 = \alpha^{\top} \widehat{\Qer}^* \alpha \le \|\alpha\|_1^2$ and $\phi^*$ is the feature mapping associated to $\widehat{\kappa}^*$, the third inequality employs the Cauchy-Schwarz inequality, the fourth inequality utilizes the Jensen's inequality.
\end{proof}
	
	\subsection{Proof of Lemma \ref{lem:3.4}}\label{append:subsec-lemm3}
The proof of Lemma \ref{lem:3.4} leverages the following two lemmas whose proofs are given in Subsections \ref{append:subsec:proof-lemma4} and  \ref{append:subsec:proof-lemma5}, respectively. 
	
\begin{lemma}\label{lem:2.2}
Suppose the system noise is modeled by the depolarization channel $\mathcal{N}_p$ in Eqn.~\eqref{eq:dep}. The noisy quantum kernel $\widehat{\Qer}$ and the ideal quantum kernel $\Qer$ has the following relation, i.e.,
		\begin{equation}
		    \Pr \left( \left|{\Qer}_{ij}-\widehat{\Qer}_{ij}\right| \ge \frac{\delta}{2} + {p}\left(1+\frac{1}{2^{N+1}}\right) \right) \le
		    \exp{\left(-\frac{\delta^2 m} {2}\right)}
		\end{equation}
	\end{lemma}	
	
	\begin{lemma}\label{lem:2.3}
		Let $\left \| \cdot \right \|$ be a given matrix norm and suppose $A,B \in \mathbb{R}^{n \times n}$ are  nonsingular and satisfy $\left \|A^{-1}(B-A)\right \| \le 1$, then
		\begin{equation}
		  \left \|  A^{-1}-B^{-1} \right \| \le \frac{\left \| A^{-1}  \right \|^2 \left \|  A-B \right \|}{1- \left \| A^{-1}(A-B) \right \|}. 
		  \end{equation}
	\end{lemma}
	
We are now ready to present the proof of Lemma \ref{lem:3.4}. 
	
	\begin{proof}[Proof of Lemma \ref{lem:3.4} ]
Recall that Lemma \ref{lem:3.4} concerns the difference between the noisy and ideal quantum kernels evaluated by  $\| {\Qer}^{-1} - \widehat{\Qer}^{-1} \|_2$. Supported by the result of Lemma \ref{lem:2.3}, the upper bound of the term $\|\Qer^{-1} -\widehat{\Qer}^{-1} \|_2$ is  
\begin{equation}
	\left \| \Qer^{-1}  - \widehat{\Qer}^{-1} \right \|_2 \le  \frac{c_{\Qer}^2\left \| \Qer  - \widehat{\Qer} \right \|_2}{1-c_{\Qer} \left \| \Qer  - \widehat{\Qer} \right \|_2},
\end{equation}
where $c_{\Qer} =  \|\Qer^{-1} \|_2 $.  
In other words, to achieve $\|\Qer^{-1} -\widehat{\Qer}^{-1} \|_2 \leq \delta$, it is sufficient to show
\begin{equation}\label{eqn:lem3-01}
	\left \|\Qer -\widehat{\Qer} \right \|_2 \le \frac{\delta} {c_{\Qer}(\delta + c_{\Qer})}.
\end{equation}
 In the following, we leverage the concentration inequality to quantify the probability when $\| {\Qer} - \widehat{\Qer}  \|_2 \geq \delta'$. Mathematically, we have
\allowdisplaybreaks
\begin{align}\label{eq:C12}
 & \Pr\left(\left \| {\Qer} - \widehat{\Qer} \right \|_2  \ge \delta' \right)
			\nonumber \\
\le & \Pr\left(\left \| {\Qer} - \widehat{\Qer} \right \|_{\Fnorm} \ge \delta' \right) 
			 \nonumber \\
= & \Pr\left(\sum_{i=1}^{n}\sum_{j=1}^{n}\left|\Qer_{ij} - \widehat{\Qer}_{ij}\right|^2 \ge \delta'^2 \right)
			  \nonumber \\
\le & \Pr\left(\bigcup\limits_{i=1}^{n}\bigcup\limits_{j=1}^{n}\left|\Qer_{ij} - \widehat{\Qer}_{ij}\right|^2 \ge \frac{\delta'^2}{n^2} \right) 
			  \nonumber \\
\le & \sum_{i=1}^{n}\sum_{j=1}^{n} \Pr \left(\left|\Qer_{ij} - \widehat{\Qer}_{ij}\right| \ge   {p}\left(1+\frac{1}{2^{N+1}}\right) + \frac{\delta'}{n}   - {p}\left(1+\frac{1}{2^{N+1}}\right) \right)
			  \nonumber \\
\le & 2n^2 \exp{\left(-2\left( \frac{\delta'}{n} - {p}\left(1+\frac{1}{2^{N+1}}\right)  \right)^2 m \right)}, 
\end{align}
where the first inequality uses
 $ \| {\Qer} - \widehat{\Qer}  \|_2 \le  \| {\Qer} - \widehat{\Qer}  \|_{\Fnorm}$, the first equality employs the definition of the Frobenius norm, the second inequality utilizes the union bound, the last second inequality is supported by the sub-additivity of probability measure, and the last inequality exploits the results of Lemma \ref{lem:2.2}. Note that to use Lemma \ref{lem:2.2}, we require
 \begin{equation}\label{eqn:lem3-02}
 	\delta' > np\left(1+\frac{1}{2^{N+1}}\right).
 \end{equation}

With setting $\delta'=\frac{\delta} {c_{\Qer}(\delta + c_{\Qer})}>  np\left(1+\frac{1}{2^{N+1}}\right)$ as indicated in Eqn.~\eqref{eqn:lem3-01} and Eqn.~\eqref{eqn:lem3-02},  we obtain
\begin{align}\label{eq:C13}
& \Pr\left(\left \| {\Qer}^{-1} - \widehat{\Qer}^{-1} \right \|_2 \ge \delta\right)  \nonumber \\
		   \le & \Pr \left(\left \|{\Qer} - \widehat{\Qer} \right \|_2  \ge \frac{\delta}{c_{\Qer}(\delta + c_{\Qer}) } \right)  \nonumber \\
	\le & 2 n^2 \exp{\left(- 2\left( \frac{\delta}{n c_{\Qer}(\delta + c_{\Qer})} - {p}\left(1+\frac{1}{2^{N+1}}\right)  \right)^2 m \right)}.	   
\end{align}

We remark that the condition $\frac{\delta} {c_{\Qer}(\delta + c_{\Qer})}>  np\left(1+\frac{1}{2^{N+1}}\right)$ is guaranteed when
\begin{equation}\label{eq:C14-2}
 \delta>\left( \left(c_{\Qer}^2np\left(1+\frac{1}{2^{N+1}}\right)  \right)^{-1} - c_{\Qer}^{-1}  \right)^{-1}. 	
\end{equation}

We now demonstrate that with probability at least $1-\delta^{\prime}/{2}$, the following relation is satisfied, i.e.,
\begin{equation}\label{eq:C15}
            \left \| {\Qer}^{-1} - \widehat{\Qer}^{-1} \right \|_2 \le   
            \left( \max \left(   \frac{\sqrt{m}}{n}c_3 - c_{\Qer}^{-1} , 0\right) \right)^{-1},
\end{equation}
where $c_3= c_{\Qer}^{-2} \left( \left(\frac{1}{2}\log \left(\frac{4n^2}{\delta^{\prime}} \right)\right)^{\frac{1}{2}} + m^{\frac{1}{2}}p\left(1+\frac{1}{2^{N+1}} \right) \right)^{-1}$. The above equation is derived by setting the last term in Eqn.~\eqref{eq:C13} as $\delta'/2$. Mathematically, we have  	
\allowdisplaybreaks	
\begin{align}\label{eq:C17}
		   & 2 n^2 \exp{\left(-2 \left(\frac{\delta}{n c_{\Qer}\left(\delta + c_{\Qer}\right)} -p\left(1+\frac{1}{2^{N+1}} \right) \right)^2 m \right)} =\frac{\delta^{\prime}}{2}  
		   \nonumber \\
	\Longrightarrow  \quad & \frac{\delta}{n c_{\Qer}\left(\delta + c_{\Qer}\right)}   = \sqrt{\frac{1}{2m}\log\left(\frac{4n^2}{\delta^{\prime}} \right)} + {p}\left(1+\frac{1}{2^{N+1}}\right) 
		\nonumber \\
		\Longrightarrow \quad & \frac{n c_{\Qer}\left(\delta + c_{\Qer}\right)}{\delta} = \left( \sqrt{\frac{1}{2m}\log\left(\frac{4n^2}{\delta^{\prime}} \right)} + {p}\left(1+\frac{1}{2^{N+1}}\right)   \right)^{-1} 
		\nonumber \\
		\Longrightarrow \quad & \frac{1}{\delta}=\frac{1}{n c_{\Qer}^2} \left( \sqrt{\frac{1}{2m}\log\left(\frac{4n^2}{\delta^{\prime}} \right)}  + {p}\left(1+\frac{1}{2^{N+1}}\right)  \right)^{-1} - c_{\Qer}^{-1}
		\nonumber \\
		\Longrightarrow \quad & {\delta}=\left( \frac{1}{n c_{\Qer}^2}\left( \sqrt{\frac{1}{2m}\log\left(\frac{4n^2}{\delta^{\prime}} \right)}  + {p}\left(1+\frac{1}{2^{N+1}}\right)  \right)^{-1} - c_{\Qer}^{-1}  \right)^{-1} \nonumber \\
		\Longrightarrow \quad & {\delta}=\left(   \frac{\sqrt{m}}{n}c_3 - c_{\Qer}^{-1}  \right)^{-1}.
		\end{align}
		
Note that the physical meaning of $\delta$ is distance, which requires $\delta>0$. Equivalently,  $\delta= \left( \max \left(   \frac{\sqrt{m}}{n}c_3 - c_{\Qer}^{-1} , 0\right) \right)^{-1}$. Meanwhile, the maximum operation naturally secures the condition required in Eqn.~\eqref{eq:C14-2}. 

	\end{proof}
	
	\subsection{Proof of Lemma \ref{lem:2.2}}\label{append:subsec:proof-lemma4}

	\begin{proof}[Proof of Lemma \ref{lem:2.2}]
Supported by the Chernoff-Hoeffding bound \cite{vapnik1992principles}, the discrepancy between $\widehat{\Qer}_{ij}$ and  $\widetilde{\Qer}_{ij}$ yields
\begin{equation}\label{eq:2.4}
			\Pr\left( \left|\widehat{\Qer}_{ij} - \widetilde{\Qer}_{ij}\right| \ge \frac{\delta}{2} \right) \le 2 \exp\left(\frac{-\delta^{ 2}m}{2} \right),
\end{equation}
 where $m$ represents the number of measurements. Moreover, the distance between the ideal result $\Qer_{ij}$ and the expectation value $\widetilde{\Qer}_{ij}=\Tr(\mathcal{N}_p(\rho(\bm{x}^{(i)}) \rho(\bm{x}^{(j)})))$ corresponding to the depolarization noise follows
		\begin{equation}\label{eq:2.5}
			\left|\Qer_{ij}-\widetilde{\Qer}_{ij}\right| \le {p} \left(\Qer_{ij} + \frac{1}{2^{N+1}}\right),
		\end{equation} 
		where ${p}$ refers to the depolarization rate in Eqn.~\eqref{eq:dep}. 
	
In conjunction with Eqn.~\eqref{eq:2.4} and \eqref{eq:2.5}, with probability at least $1-2\exp{\left(-\frac{\delta^{2}m}{2} \right) }$, we have
		\begin{equation}\label{eq:lem4-1}
			\left|\widehat{\Qer}_{ij} - \Qer_{ij} \right|= \left|\widehat{\Qer}_{ij}-\widetilde{\Qer}_{ij} +\widetilde{\Qer}_{ij}- \Qer_{ij} \right| \le {p} \left(\Qer_{ij} + \frac{1}{2^{N+1}}\right) + \frac{\delta}{2} \le {p} \left(1 + \frac{1}{2^{N+1}}\right) + \frac{\delta}{2}.
		\end{equation}
After simplifying Eqn.~\eqref{eq:lem4-1}, we achieve 
		\begin{equation}
			\Pr\left(\left|\Qer_{ij}-\widehat{\Qer}_{ij}\right| \ge {p} \left(1 + \frac{1}{2^{N+1}}\right) + \frac{\delta}{2} \right) \le 2 \exp\left(\frac{-\delta^{ 2}m}{2} \right).
		\end{equation}
\end{proof}
	
	\subsection{Proof of Lemma \ref{lem:2.3}}\label{append:subsec:proof-lemma5}
	
	\begin{proof}[Proof of Lemma \ref{lem:2.3}]
	Note that the difference between $A^{-1}$ and $B^{-1}$ satisfies $A^{-1}-B^{-1}=A^{-1}(B-A)B^{-1}$. This relation implies 
\begin{equation}\label{eq:2.8}
			\left \|A^{-1}-B^{-1}\right  \|= \left \|A^{-1}(B-A)B^{-1} \right \| \le \left \|A^{-1}(B-A)\right \| \left \|B^{-1}\right\|.
\end{equation}
Moreover, since $B^{-1} = A^{-1}-A^{-1}(B-A)B^{-1} $, the term is upper bounded by		
\begin{align}\label{eq:2.9} 	
	& \left\|B^{-1} \right \| \le \left\|A^{-1}\right \| + \left\|A^{-1}(B-A)B^{-1}\right \| \le \left\|A^{-1}\right \| + \left\|A^{-1}(B-A)\right \| \left\|B^{-1}\right \|  \nonumber\\
\Rightarrow  &	\left \|B^{-1}\right \| \le \frac {\left\|A^{-1}\right\|} {1-\left\|A^{-1}(A-B)\right\|}.
\end{align}
Combining Eqn.~\eqref{eq:2.8} with Eqn.~\eqref{eq:2.9}, we achieve  
\begin{equation}
\left \|A^{-1}-B^{-1}\right \| \le \frac{\left\|A^{-1} \right \| \left\|A^{-1}(A-B)\right \|} {1-\left\|A^{-1}(A-B)\right \|}  \le \frac{\left\|A^{-1} \right \|^2 \left\|A-B\right \|} {1-\left\|A^{-1}(A-B)\right \|}.	
\end{equation}
		
	\end{proof}

\section{The comparison of the generalization error bound between ideal and noisy quantum kernels }   
\label{append:sec:saturation} 

In this section, we demonstrate that the achieved upper bound for noisy quantum kernels in Theorem~\ref{thm:3.1} is non-trivial, where the generalization error bound of ideal quantum kernels \cite{huang2021power} can not be used to analyze the generalization ability of noisy quantum kernels. In the following, we provide both the theoretical and numerical evidence to support our claim.

\subsection{Numerical evidence for the saturation}
 \label{append:subsec:saturation_numerical} 

Here we conduct numerical simulations to exhibit how the the generalization error (i.e., prediction accuracy) varies with respect to the different size of training dataset $n$ under ideal and noisy scenarios, respectively.  The hyper-parameters settings are as follows.
    The size of training dataset and the number of measurements are set as $n \in \{5, 50, 100, 200 \}$ and $m \in \{10, 100, 500, 1000, \text{inf}\}$ respectively. The number of qubits used to establish quantum kernels is set as $N=2$. When the noisy scenario is considered, the depolarizing rate is set as $\tilde{p}=0.05$. 
    
The simulation results of the ideal quantum kernel are shown in Figure \ref{fig:round2_fig1}, highlighted by the solid blue line (i.e., $\tilde{p}=0$ \& $ m=\mbox{inf}$). Specifically, the prediction accuracy for the ideal kernel keeps on increasing from $79\%$, $94\%$, $95\%$ to $96\%$ when the size of the training dataset $n$ increases from $5$, $50$, $100$, to $200$. The achieved prediction accuracy indicates that a larger size of the training dataset ensures a better generalization ability. These observations exactly echo with the conclusion of Ref.~\cite{huang2021power}, where the generalization error bound for the ideal kernel scales with $\tilde{O}(\sqrt{c_1/n})$ in Eqn.~\eqref{eqn:gene_error_ideal}.

Different from the ideal setting (i.e., $\tilde{p}=0$ \& $m=\mbox{inf}$), the prediction accuracy achieved by noisy kernels, highlighted by the dotted lines, contrasts with Huang's upper bound $\tilde{O}(\sqrt{c_1/n})$ in Eqn.~\eqref{eqn:gene_error_ideal}. In particular, for the case of $m \in \{10, 100, 500\}$ ($m=1000$), the prediction accuracy for the noisy kernel reaches the peak when $n=50$ ($n=100$) and then begins to decline, while Huang's result claims that the maximum prediction accuracy should be $n=200$. The obtained empirical results accord with our results in Eqn.~\eqref{eqn:gene_err_noisy} but contradicts with Huang's results in Eqn.~\eqref{eqn:gene_error_ideal} such that increasing $n$ for the noisy kernel may decrease the prediction accuracy, or equivalently increase the generalization error. In other words, the achieved generalization error bound for noisy kernels is by no means trivially above Huang’s upper bound.

\begin{figure}[h!]
    \centering
    \includegraphics[width=0.60\textwidth]{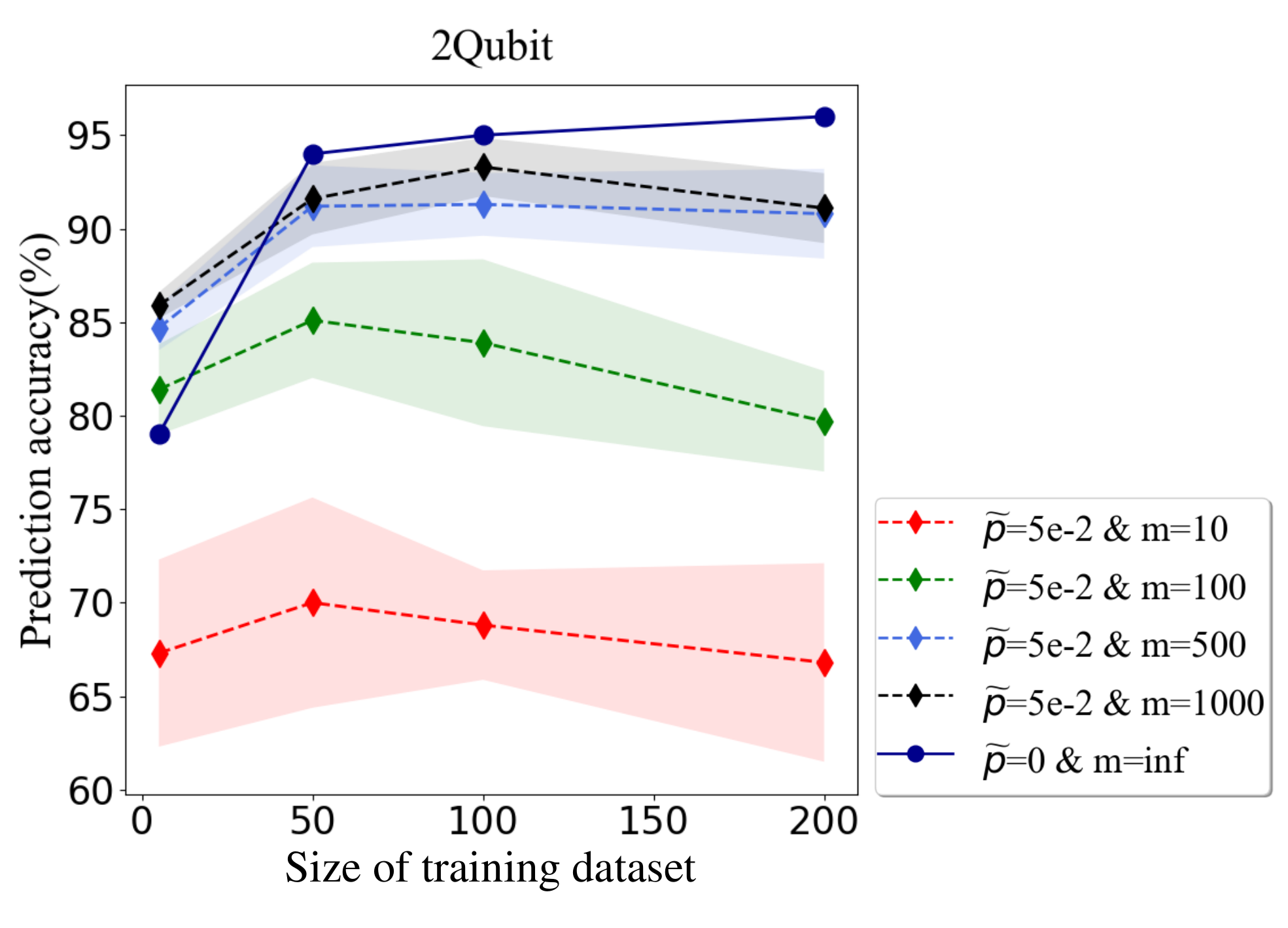}
    \caption{\small{\textbf{The different generalization behavior for noisy and ideal quantum kernel.} The solid line and dotted line refer to the prediction accuracy of the ideal kernel and noisy kernel over varied size of training dataset $n \in \{5, 50, 100, 200\}$ when the number of qubit $N=2$ and the depolarizing rate $\tilde{p}=0.05$, respectively.  }}
	\label{fig:round2_fig1}
    \end{figure}

    \subsection{Theoretical evidence for saturation}
    \label{append:subsec:saturation_theoretical} 

Envisioned by the aforementioned simulation results, we now theoretically investigate that the generalization error bound of noisy kernels must contain  a  term  that is proportional  to $n$, e.g., $\tilde{\Omega}\left( \sqrt{c/n} + \sqrt{\frac{1}{c} \frac{n}{\sqrt{m}}} \right)$. This is equivalent to explore whether the generalization error bound in Theorem \ref{thm:3.1} is saturated. Notably, according to Eqn.~\eqref{eqn:gene_err_noisy}, the achieved upper bound is constituted by two terms, i.e., $\sqrt{c_1/n}$ and $\sqrt{n/(c_2\sqrt{m})}$. In particular, the former quantifies the generalization error in the ideal setting and the latter origins from the kernel approximation error between $\Qer$ and $\widehat{\Qer}$. With this regard, in the following, we will separately analyze the saturation of the achieved generalization error bound in the ideal and NISQ settings.

   \noindent\textit{\underline{The saturation in the ideal setting.}} The term $\sqrt{c_1/n}$ refers to the generalization error upper  bound in the ideal setting, i.e., the quantum system is noiseless and the number of measurements is infinite ($m\rightarrow \infty$). As shown in \cite{huang2021power} (see Section D~2, Page 17), the generalization error bound of quantum kernels is obtained by deriving the upper bound of the empirical Rademacher complexity $\widehat{\mathfrak{R}}_{\mathcal{D}}(\mathcal{H})$, where $\mathcal{H}=\{ h\in \mathcal{H}: \| h\|_{\mathbb{H}} \le \Lambda \}$ is the hypothesis set, $\mathbb{H}$ is the reproducing kernel Hilbert space (RKHS) of the quantum kernel function of $\Qer$, $h$ is quantum kernel based hypothesis, and $\Lambda$ refers to a  bounded constant. Notably, the term $\sqrt{c_1/n}$ refers to the upper bound of $\widehat{\mathfrak{R}}_{\mathcal{D}}(\mathcal{H})$, i.e.,
\begin{equation}\label{eqn:p1.1-1}
	\widehat{\mathfrak{R}}_{\mathcal{D}}(\mathcal{H}) \leq \sqrt{\frac{c_1}{n}}.
\end{equation} 
The study \cite{mohri2018foundations}  has proven that the term $\sqrt{c_1/n}$ is saturated, supported by the following theorem. 
\begin{theorem}
    [Modified from theorem 6.12, \cite{mohri2018foundations}]\label{thm:2} Let $\kappa: \mathcal{X} \times \mathcal{X} \to \mathbb{R}$ be a positive definite symmetric kernel and let $\Phi: \mathcal{X} \to  \mathbb{H}$ be a feature mapping associated to $\kappa$. Let $\mathcal{D} \subset \{\bm{x}: \kappa(\bm{x}, \bm{x}) \le r^2\}$ be a sample of size $n$, and let $\mathcal{H} = \{\bm{x} \mapsto \braket{\bm{\omega}, {\Phi}(\bm{x})}: \| \bm{\omega} \|_{\mathbb{H}} \le \Lambda = \sqrt{c_1} \}$ for some $c_1 \ge 0$. Without loss of generality, the corresponding kernel matrix $\Qer$ is   assumed to satisfy $\Tr(\Qer)=n$. Then the empirical Rademacher complexity yields
    \begin{equation}
        \frac{1}{\sqrt{2}}\sqrt{\frac{c_1}{n}} = \frac{1}{\sqrt{2}}\frac{\Lambda\sqrt{\Tr(\Qer)}}{n} \le \widehat{\mathfrak{R}}_{\mathcal{D}}(\mathcal{H}) \le \frac{\Lambda\sqrt{\Tr(\Qer)}}{n} = \sqrt{\frac{c_1}{n}}.
    \end{equation}
\end{theorem}

The results in Theorem \ref{thm:2} indicate that the upper and lower bounds of the empirical Rademacher complexity $\widehat{\mathfrak{R}}_{\mathcal{D}}(\mathcal{H})$ follow the same scaling behavior, which is $\sqrt{c_1/n}$.   To this end, we can conclude that the first term $\sqrt{c_1/n}$ in Theorem \ref{thm:3.1} is saturated.

\textit{\underline{The saturation in the noise setting.}} 
 We then analyze the saturation of the achieved generalization error bound in the NISQ scenario. This part also serves as the proof of Theorem \ref{thm:3.1}. Notably, compared with the ideal case, there is an additional term $\sqrt{n/(c_2\sqrt{m})}$  introduced by the quantum system noise and sample error. In other words, understanding the saturation of the achieved generalization error bound for noisy quantum kernels is equivalent to quantifying the saturation of $\sqrt{n/(c_2\sqrt{m})}$. Recall that this term is originated from the kernel approximation error between $\widehat{\Qer}^{-1}$ and  $\Qer^{-1}$, i.e., \begin{equation}\label{eqn:point-1-1-2}
 	\sqrt{\|\widehat{\Qer}^{-1} -  {\Qer}^{-1} \|_2} \leq \sqrt{\frac{n}{c_2\sqrt{m}}}.
 \end{equation}
 Under the above observation, here we intend to examine  the saturation of the kernel approximation error between $\widehat{\Qer}^{-1}$ and  $\Qer^{-1}$,  especially for the dependence of $\sqrt{n}$.

  For ease of understanding, we only consider the sampling error (i.e., a finite number of measurement $m$) and set the depolarizing rate as $p=0$. In other words, the element of the noisy quantum kernel yields $\widehat{\Qer}_{ij} = \frac{1}{m}\sum_{k=1}^m V_k,~\forall i,j\in[n],$ where $V_k \sim \Ber(W_{ij})$. In this scenario, with probability  $\delta^{\prime}>0$, there exists a noisy kernel satisfying  
\begin{align*}
    \left(\widehat{\Qer}^{-1}\right)_{ij} =  \left({\Qer}^{-1}\right)_{ij} \pm \epsilon, \quad ~\forall i,j\in[n],
\end{align*}
where $\epsilon>0$ refers to the random error. An immediate observation is that the spectral norm of the difference between $\Qer^{-1}$ and $\widehat{\Qer}^{-1}$ satisfies
\begin{equation} \label{eq:lower_nound}
    \sqrt{ \left\|\Qer^{-1}-  \widehat{\Qer}^{-1} \right\|_2} \ge \sqrt{\frac{1}{\sqrt{n}} \left\|\Qer^{-1}-  \widehat{\Qer}^{-1} \right\|_{\Fnorm}} = \sqrt{\frac{1}{\sqrt{n}}  n \epsilon} =\sqrt{\sqrt{n} \epsilon},
\end{equation}
where the first inequality employs   $\|\cdot \|_2 \ge \|\cdot \| _{\Fnorm}/\sqrt{n} $. 

In conjunction with Eqn.~\eqref{eqn:point-1-1-2} and Eqn.~\eqref{eq:lower_nound}, we obtain 
\begin{equation}
	\sqrt{\sqrt{n} \epsilon}	\leq \sqrt{\|\widehat{\Qer}^{-1} -  {\Qer}^{-1} \|_2} \leq \sqrt{\frac{n}{c_2\sqrt{m}}}.
\end{equation}
The achieved lower and upper bounds in the above equation provide two implications.
First, the generalization error bound of noisy kernels must contain a term that is proportional to $n$. Second, the second term in Theorem \ref{thm:3.1} is nearly saturated, where the lower bound and our upper bound are separated by a scaling factor $n^{1/4}$.

To summarize, the generalization error bound in Theorem \ref{thm:3.1} is exactly saturated in the ideal case. Besides, the generalization error bound in Theorem \ref{thm:3.1} is nearly saturated in the NISQ setting, where the lower and upper bounds are separated by a factor $n^{1/4}$.
Therefore, the generalization error bound of noisy kernels must contain a term that is proportional to $n$. In other words, the bound in Eqn.~\eqref{eqn:gene_error_ideal} fails to explain the generalization ability of noisy quantum kernels.

	\section{Proof of Lemma \ref{lem:proj-kernel}}\label{append:proj-psd-kernel}
In this section, we introduce three spectral transformation techniques, i.e., the clipping, flipping, and shifting methods, which are utilized to enhance generalization performance of noisy quantum kernels. Moreover, we provide the theoretical evidence such that these transformations methods enable a better generalization performance.

\subsection{Spectrum clipping method}\label{append:subsec:clip}
The construction rule of the noisy quantum kernel implies that $\widehat{\Qer}$ in Eqn.~\eqref{eqn:est-ele-ker} is symmetric and thus it has an eigenvalue decomposition 
\begin{equation}
	\widehat{\Qer} = U^{\top} \Lambda U,
\end{equation}
where $U$ is an orthogonal matrix and $\Lambda$ is a diagonal matrix of real eigenvalues with $\Lambda = \diag(\widehat{\lambda}_1, \cdots, \widehat{\lambda}_n)$. Without loss of generality, we assume that $\widehat{\lambda}_1 \ge \cdots \ge \widehat{\lambda}_r \ge 0 \ge \widehat{\lambda}_{r+1} \ge \cdots \ge \widehat{\lambda}_n.$ 

The mechanism of the spectrum clip method is clipping all the negative eigenvalues of $\widehat{\Qer}$ to zero. Intuitively, the negative eigenvalues of $\widehat{\Qer}$ have been regarded as the result of noise disturbance, where setting them to zero is treated as the denoising step \cite{wu2005analysis}. Define the clipped eigenvalues as
	\begin{equation}
		\Lambda_{\clip} = \diag\left ( \widehat{\lambda}_1, \cdots, \widehat{\lambda}_r, 0, \cdots, 0 \right).
	\end{equation}
The calibrated  quantum kernel yields 
\begin{equation}\label{eq:4.1}
\widehat{\Qer}_{\clip} = U^{\top} \Lambda_{\clip} U.	
\end{equation}
 The following lemma exhibits that the discrepancy between the calibrated quantum kernel and the ideal quantum kernel is lower than  the difference between the original noisy quantum kernel and the ideal quantum kernel.	
	\begin{lemma}\label{lem:4.1}
	Given a noisy quantum kernel $\widehat{\Qer}$ in Eqn.~\eqref{eqn:est-ele-ker},  let $\widehat{\Qer}_{\clip}$ in Eqn.~\eqref{eq:4.1} be the  quantum kernel calibrated by the spectrum clipping method. The distance of the ideal kernel $\Qer$, the noisy quantum kernel $\widehat{\Qer}$, and the calibrated quantum kernel $\widehat{\Qer}_{\clip}$ satisfies 
		\begin{equation} \label{eq:m_clip}
			\left \| \Qer - \widehat{\Qer}_{\clip} \right \|_{\Fnorm} \le \left \| \Qer - \widehat{\Qer} \right \|_{\Fnorm}. 
		\end{equation}
	\end{lemma}
	
\begin{proof}[Proof of Lemma \ref{lem:4.1}]
We note that supported by the definition of the Frobenius norm,   an equivalent of achieving Eqn.~\eqref{eq:m_clip} is
\begin{equation}\label{eqn:lem-clip-core}
	 \Tr \left(  \left( \Qer - \widehat{\Qer}_{\clip}   \right)^2 \right)   \le
    \Tr \left(  \left( \Qer - \widehat{\Qer}   \right)^2 \right).
\end{equation} 
After simplification, Eqn.~\eqref{eqn:lem-clip-core} can be rewritten as
\begin{equation}\label{eqn:lem-clip-core2}
	 \Tr \left(  \widehat{\Qer}_{\clip}^2 \right) -2 \Tr \left(  {\Qer}\widehat{\Qer}_{\clip} \right)     \le
    \Tr \left( \widehat{\Qer}^2 \right) -2 \Tr \left(  {\Qer}\widehat{\Qer} \right).
\end{equation}

To achieve Eqn.~\eqref{eqn:lem-clip-core2}, we now prove  we only need to prove $\Tr(\widehat{\Qer}_{\clip}^2) \le \Tr(\widehat{\Qer}^2 )$ and $\Tr(\Qer \widehat{\Qer}_{\clip}) \ge \Tr(\Qer \widehat{\Qer} )$ separately. Following notation in the main text, we denote $\widehat{\Qer} = \sum_{i=1}^n \lambda_i \bm{u}_i \bm{u}_i^{\top}$ and $\widehat{\Qer}_{\clip} = \sum_{i=1}^r \lambda_i \bm{u}_i \bm{u}_i^{\top}$.
		
The relation $\Tr(\widehat{\Qer}_{\clip}^2) \le \Tr(\widehat{\Qer}^2 )$ can be efficiently obtained based on the definition of trace operation. Specifically, we have 
		\begin{equation}\label{eqn:lem1-clip-core3}
			\Tr\left(\widehat{\Qer}_{\clip}^2 \right) = \sum_{i=1}^r \widehat{\lambda}_i^2 \le \sum_{i=1}^n \widehat{\lambda}_i^2 = \Tr\left(\widehat{\Qer}^2\right).
		\end{equation}
The linear property of the trace operation allows us to achieve $\Tr(\Qer \widehat{\Qer}_{\clip}) \ge \Tr(\Qer \widehat{\Qer} )$, i.e.,  
		\begin{align}\label{eq:4.4}
			& \Tr\left(\Qer \widehat{\Qer} \right) \nonumber\\
			=  &  \Tr\left( \left(\sum_{i=1}^n \lambda_i \bm{v}_i \bm{v}_i^{\top} \right) \left(\sum_{j=1}^n \widehat{\lambda}_j \bm{u}_j \bm{u}_j^{\top} \right) \right)  \nonumber \\
		=	 &  \Tr\left( \left(\sum_{i=1}^n \lambda_i \bm{v}_i \bm{v}_i^{\top} \right) \left(\sum_{j=1}^r  \widehat{\lambda}_j\bm{u}_j \bm{u}_j^{\top} +\sum_{j=r+1}^n                    \widehat{\lambda}_j \bm{u}_j \bm{u}_j^{\top} \right) \right) \nonumber \\
			= & \Tr \left(\Qer \widehat{\Qer}_{\clip} \right) + \sum_{i=1}^n                 \sum_{j=r+1}^n \lambda_i \widehat{\lambda}_j \left (\bm{u}_j^{\top} \bm{v}_i \right)^2        \nonumber \\
			\le  &  \Tr \left(\Qer \widehat{\Qer}_{\clip} \right),
		\end{align}
		where the last equality uses the fact that $\lambda_i \ge 0$ and $\widehat{\lambda}_j \le 0 $ for $\forall i \in [n], \forall j \in [n] \backslash [r]$. 
	
In conjunction with Eqns.~(\eqref{eqn:lem-clip-core2}), (\eqref{eqn:lem1-clip-core3}), and (\eqref{eq:4.4}), we obtain 
	\begin{equation}
		 \left \| \Qer - \widehat{\Qer}_{\clip} \right \|_{\Fnorm}^2 \le \left \| \Qer - \widehat{\Qer} \right\|_{\Fnorm}^2.
	\end{equation}	
	\end{proof}

\subsection{Spectrum flipping method}\label{append:subsec:flip}
The spectral flipping method is proposed by \cite{pekalska2001generalized}. Different from the clipping method \cite{wu2005analysis} that interpreting the negative eigenvalues is caused by noise, Refs. \cite{laub2004feature} and  \cite{laub2006information} showed that the negative eigenvalues in $\widehat{\Qer}$ may encode useful information about data features or categories.  Following notation in Subsection \ref{append:subsec:clip},   the spectral flipping method flips the sign of the negative eigenvalues of  $\widehat{\Qer}$ to obtain the calibrated quantum kernel 
\begin{equation}\label{eqn:flip-kernel}
	\widehat{\Qer}_{\flip} = U^{\top} \Lambda_{\flip} U,
\end{equation}
where $\Lambda_{\flip}= \diag \left(\widehat{\lambda}_1, \cdots, \widehat{\lambda}_r, -\widehat{\lambda}_{r+1}, \cdots, -\widehat{\lambda}_{n} \right)$.

\begin{lemma}\label{lem:4.2}
Given a noisy quantum kernel $\widehat{\Qer}$ in Eqn.~\eqref{eqn:est-ele-ker},  let $\widehat{\Qer}_{\flip}$ in Eqn.~\eqref{eqn:flip-kernel} be the  quantum kernel calibrated by the spectrum flipping method. The distance of the ideal kernel $\Qer$, the noisy quantum kernel $\widehat{\Qer}$, and the calibrated quantum kernel $\widehat{\Qer}_{\flip}$ satisfies 
\begin{equation}\label{eqn:core-flip1}
			\left\| \Qer - \widehat{\Qer}_{\flip} \right\|_{\Fnorm} \le \left\| \Qer - \widehat{\Qer} \right\|_{\Fnorm}. 
		\end{equation}
	\end{lemma}
	
	\begin{proof}[Proof of Lemma \ref{lem:4.2}]
This proof of Lemma \ref{lem:4.2} is analogous to Lemma \ref{lem:4.1}. In particular, we attain Eqn.~\eqref{eqn:core-flip1} via proving $\Tr(\widehat{\Qer}_{\flip}^2) \le \Tr(\widehat{\Qer}^2 )$ and $\Tr(\Qer \widehat{\Qer}_{\flip}) \ge \Tr(\Qer \widehat{\Qer} )$ separately. Note that the calibrated quantum kernel yields the following spectral decomposition, i.e., $\widehat{\Qer}_{\flip} = \sum_{i=1}^n \widehat{\lambda}_i \bm{u}_i \bm{u}_i^{\top} - \sum_{i=r+1}^n \widehat{\lambda}_i \bm{u}_i \bm{u}_i^{\top} $. An immediate observation is that 
		\begin{equation}\label{eqn:core-flip2}
			\Tr\left(\widehat{\Qer}_{\flip}^2 \right) = \sum_{i=1}^n \widehat{\lambda}_i^2 = \Tr\left(\widehat{\Qer}^2 \right).
		\end{equation}
Moreover, the relation between  $ \Tr(\Qer \widehat{\Qer}) $ and $\Tr (\Qer \widehat{\Qer}_{\flip})$ follows
		\begin{align}\label{eqn:core-flip3}
			& \Tr\left(\Qer \widehat{\Qer}\right) \nonumber\\
			 = & \Tr\left( \left(\sum_{i=1}^n \lambda_i \bm{v}_i \bm{v}_i^{\top} \right) \left(\sum_{j=1}^n \widehat{\lambda}_j \bm{u}_j \bm{u}_j^{\top} \right) \right)  \nonumber \\
			  = & \Tr\left( \left(\sum_{i=1}^n \lambda_i \bm{v}_i \bm{v}_i^{\top} \right) \left(\sum_{j=1}^r \widehat{\lambda}_j \bm{u}_j \bm{u}_j^{\top} - \sum_{j=r+1}^n \widehat{\lambda}_j \bm{u}_j \bm{u}_j^{\top} + 2 \sum_{j=r+1}^n \widehat{\lambda}_j \bm{u}_j \bm{u}_j^{\top} \right)   \right)  \nonumber \\
			  = & \Tr \left(\Qer \widehat{\Qer}_{\flip} \right) + 2\sum_{i=1}^n               \sum_{j=r+1}^n \lambda_i \widehat{\lambda}_j \left(\bm{u}_j^{\top}            \bm{v}_i \right)^2 \nonumber \\
			  \le & \Tr \left(\Qer \widehat{\Qer}_{\flip}\right).
		\end{align}
Combining Eqn.~\eqref{eqn:core-flip2} with Eqn.~\eqref{eqn:core-flip3}, we achieve $ \| \Qer - \widehat{\Qer}_{\flip} \|_{\Fnorm}^2 \le \| \Qer - \widehat{\Qer} \|_{\Fnorm}^2$.  
	\end{proof}
	
	\subsection{Spectrum shifting method}
The mechanism of the spectrum shifting method is shifting the spectrum of the noisy quantum kernel $\widehat{\Qer}$  by the absolute value of its minimum eigenvalue $\left|\lambda_{n} \right|$. Mathematically, the calibrated quantum kernel yields
\begin{equation}\label{eqn:shift-kernel}
	\widehat{\Qer}_{\shift} = \widehat{\Qer} + \left |\min \left(\lambda_{n} , 0 \right) \right|\mathbb{I}.
\end{equation}
Note that the spectrum shifting method ensures that any indefinite matrix $\widehat{\Qer}$ can be calibrated to be PSD. Compared with the clipping and flipping methods, the spectrum shifting method only enhances all self-similarities by the amount of $\left|\lambda_{n} \right|$ and does not change the relative similarity between any two different samples.
	
	\begin{lemma}\label{lem:4.3}
	Given a noisy quantum kernel $\widehat{\Qer}$ in Eqn.~\eqref{eqn:est-ele-ker}  satisfying $\Tr(\widehat{\Qer}) = n$,  let $\widehat{\Qer}_{\shift}$ in Eqn.~\eqref{eqn:shift-kernel} be the quantum kernel calibrated by the spectrum shifting method. The distance of the ideal kernel $\Qer$, the noisy quantum kernel $\widehat{\Qer}$, and the calibrated quantum kernel $\widehat{\Qer}_{\shift}$ satisfies 
		\begin{equation}\label{eqn:core-shift1}
		\left\| \Qer - \widehat{\Qer}_{\shift} \right\|_{\Fnorm} \le \left  \|  \Qer - \widehat{\Qer} \right\|_{\Fnorm}. 
		\end{equation}
	\end{lemma}

	\begin{proof}

The proof of Lemma~\ref{lem:4.3} can be efficiently achieved by leveraging the equivalence of Frobenius norm and trace operation, that is $\| \Qer - \widehat{\Qer}_{\shift} \|_{\Fnorm}^2=\Tr( ( \Qer - \widehat{\Qer}_{\shift} ) ^2)$. 

In particular, we first rewrite the term $\Tr( \widehat{\Qer}_{\shift}^2 )$, i.e., 
\begin{align}\label{eq:D11}
	    & \Tr\left( \widehat{\Qer}_{\shift}^2 \right) 
	    \nonumber\\
	    = & \Tr  \left(\left( \widehat{\Qer}-\widehat{\lambda}_n \mathbb{ \mathbb{I}} \right)^2 \right) \nonumber \\
	     = & \Tr\left( \widehat{\Qer}^2 \right) - 2\widehat{\lambda}_n \Tr\left( \widehat{\Qer} \right) + n\widehat{\lambda}_n^2.
	    \end{align} 
 Second, the term $\Tr( \Qer \widehat{\Qer} )$ can be rewritten as  
	    \begin{align}
	    &\Tr\left( \Qer \widehat{\Qer}  \right)
	    \nonumber \\
	    = & \Tr  \left(\Qer \left( \widehat{\Qer}-\widehat{\lambda}_n  \mathbb{I} +\widehat{\lambda}_n \mathbb{I} \right) \right) 
	    \nonumber \\
	     = & \Tr\left( \Qer \widehat{\Qer}_{\shift}  \right) +\Tr\left( \widehat{\lambda}_n \Qer  \right)  \nonumber \\
	     = & \Tr\left( \Qer \widehat{\Qer}_{\shift}  \right) + n \widehat{\lambda}_n.\label{eq:D15}
	    \end{align} 
Substituting Eqn.~\eqref{eq:D11} and Eqn.~\eqref{eq:D15} with $ \|\Qer -\widehat{\Qer}  \|_{\Fnorm}^2$, we have 
	\begin{align}
	    & \left \|\Qer -\widehat{\Qer} \right \|_{\Fnorm}^2 
	    \nonumber \\
	    = & \Tr\left( \Qer^2  \right) -2\Tr\left( \Qer \widehat{\Qer}  \right) + \Tr\left( \widehat{\Qer}^2  \right) 
	    \nonumber \\
	    = & \Tr\left( \Qer^2  \right) -2\Tr\left( \Qer \widehat{\Qer}_{\shift}  \right) + \Tr\left( \widehat{\Qer}_{\shift}^2  \right) -2n \widehat{\lambda}_n + 2\widehat{\lambda}_n \Tr\left( \widehat{\Qer} \right) - n\widehat{\lambda}^2 
	    \nonumber \\
	    = & \left \|\Qer -\widehat{\Qer}_{\shift} \right\|_{\Fnorm}^2 - n\widehat{\lambda}_n^2 
	    \nonumber \\
	    \le & \left \|\Qer -\widehat{\Qer}_{\shift} \right \|_{\Fnorm}^2,
	    \end{align}
	    where the second equality employs $\Tr( \widehat{\Qer} ) = n$.  
	    
	\end{proof}

\section{More details about numerical simulations}\label{append:num-sim}
In this section, we append more implementation details and simulation results as omitted in the main text. Specifically, we first introduce the construction rule of the employed dataset in Subsection \ref{append:subsec:dataset}. Next, we explain the detailed implementation of the quantum kernel to achieve quantum advantages in Subsection \ref{append:subsec:impt-q-kernel}. Last, in Subsection \ref{append:subsec:sim-res}, we conduct comprehensive simulations to demonstrate how the quantum system noise, the number of measurements, and dataset size influence performance of quantum kernels, and how spectral transformation techniques improve performance of quantum kernels under NISQ settings.

\subsection{The construction of the datasets}\label{append:subsec:dataset}
The construction rule of the Fashion-MNIST dataset with the reconstructed labels follows Ref. \cite{huang2021power}. Recall that after the preprocessing stage, we collect the training set $\mathcal{D}$ and the test set $\mathcal{D}_{Te}$ with $n$ and $n_{Te}$ examples, respectively. Applying the designated classical and quantum kernel functions to $\mathcal{D}\cup \mathcal{D}_{Te}$, we obtain the classical kernel $\Ker_{all}\in \mathbb{R}^{(n+n_{Te})\times (n+n_{Te})}$ and the quantum kernels $\Qer_{all}\in \mathbb{R}^{(n+n_{Te})\times (n+n_{Te})}$. As explained in Appendix \ref{append:thm1}, the superiority of quantum kernels can be achieved by seeking an optimal label vector $Y^*\in \{-1, 1\}^{n+n_{Te}}$ that maximizes their geometric difference in Eqn.~\eqref{eqn:Geo-diff}. Mathematically, the optimal solution yields 
\begin{equation}
Y^* :=  g(Y^{\sharp}) =  g(\sqrt{ \Qer}_{all} \bm{v}),
\end{equation}
where $\bm{v}$ is the eigenvector of 
$\sqrt{\Qer_{all}} {\Ker_{all}} \sqrt{ \Qer_{all}}$ corresponding to the eigenvalue $\|\sqrt{\Qer_{all}} {\Ker_{all}} \sqrt{ \Qer_{all}}\|_2$, and the function $g(\cdot)$ maps $Y^{\sharp}_i$ with `1' if $Y^{\sharp}_i>\Median(Y^{\sharp}_1, \cdots, Y^{\sharp}_{n+n_{Te}})$ and `-1' if $Y^{\sharp}_i  \le \Median(Y^{\sharp}_1, \cdots,Y^{\sharp}_{n+n_{Te}})$ for $\forall i \in [{n+n_{Te}}]$.

\subsection{The protocol of quantum kernels}\label{append:subsec:impt-q-kernel}
	We employ the variational quantum circuit architecture proposed in \cite{havlivcek2019supervised} to construct the quantum kernel function $W_{ij}= \Tr(\rho(\bm{x}^{(i)}) \rho(\bm{x}^{(j)}))$ in Eqn.~\eqref{eqn:Q-kernel}. Specifically, the quantum feature map $\rho(\bm{x}^{(i)}) = \ket{\varphi(\bm{x}^{(i)})}\bra{\varphi(\bm{x}^{(i)}})$ is defined as  $\ket{\varphi(\bm{x}^{(i)})} = U_E(\bm{x}_i)$ with
\begin{equation}
		U_E(\bm{x}_i) = U_Z(\bm{x}_i)H^{\otimes N} U_Z(\bm{x}_i) H^{\otimes N} \ket{0}^{\otimes N},
\end{equation}
	where $H^{\otimes N}$ is the unitary that applies Hadamard gates on all qubits in parallel, 
		$U_Z(\bm{x}_i) = \exp ( \sum_{j=1}^N \bm{x}_{ij} Z_j +  \sum_{j=1}^N \sum_{j'=1}^N \bm{x}_{ij}\bm{x}_{ij'} Z_j Z_{j'} )$, and $Z_j$ is the Pauli-Z operator acting on the $j$-th qubit.  We summarize the quantum circuit implementation to calculate $\Qer_{ij}$  in Figure  \ref{fig:embed}.

	\begin{figure*}[h!]
		\centering
	\includegraphics[width=0.98\textwidth]{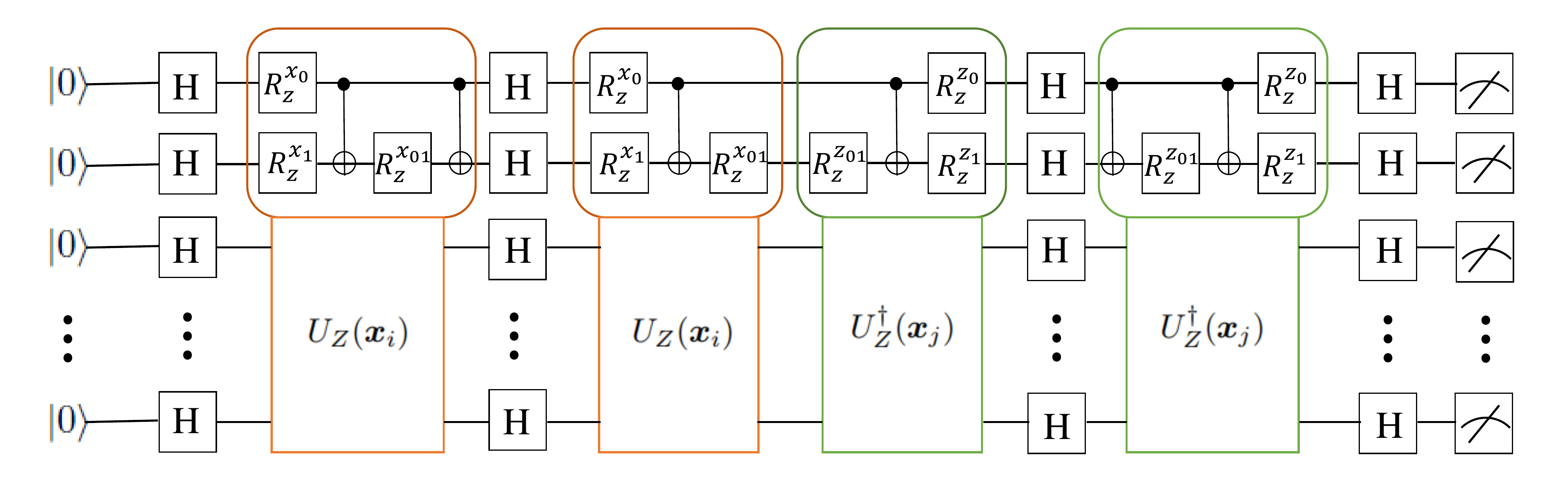}
		\caption{\small{\textbf{The quantum circuit architecture to produce the kernel element $\Qer_{ij}$.} The implementation of the quantum circuit can be separated into two parts. The first part maps the classical input $\bm{x}_i$ to the quantum feature space $\ket{\varphi(\bm{x}^{(i)})}$, highlighted by the red color. The second part maps the classical input $\bm{x}_j$ to the quantum feature space $\bra{\varphi(\bm{x}^{(j)})}$, highlighted by the green color. The expectation value of the quantum measurements   corresponds to $\Qer_{ij}$}.}
		\label{fig:embed}
	\end{figure*}
	
	\subsection{The hyper-parameter of RBF kernels and the regularization parameter}\label{append:subsec:hyper-para}
We perform a grid search over the regularization parameter to make full use the power of RBF kernels. Concretely, denoted that $\Var[x_{ik}]$ is the variance of all coordinates $k=1, \cdots, d$ from all the data points $\{\bm{x}_1, \cdots, \bm{x}_n\}\in\mathbb{R}^d$, the explicit settings of the hyper-parameter $\gamma$ of the RBF kernel and the regularization parameter are  
\begin{equation}\label{eq:gamma}
    \gamma \in \{ 0.25, 0.5, 1.0, 2.0, 4.0, 5.0, 10.0, 20.0, 40.0, 50.0 \}/
    (d \Var[x_{ik}]),
\end{equation}
and 
\begin{align}\label{eq:lambda}
    \lambda \in \{0.006, 0.015, 0.03, 0.0625, 0.125, 0.25, 0.5,  
    1, 2, 4, 8, 16, 32, 64, 128, 256, 512, 1024  \},
\end{align}
respectively.

\begin{figure}[h!]
\centering
\includegraphics[width=0.68\textwidth]{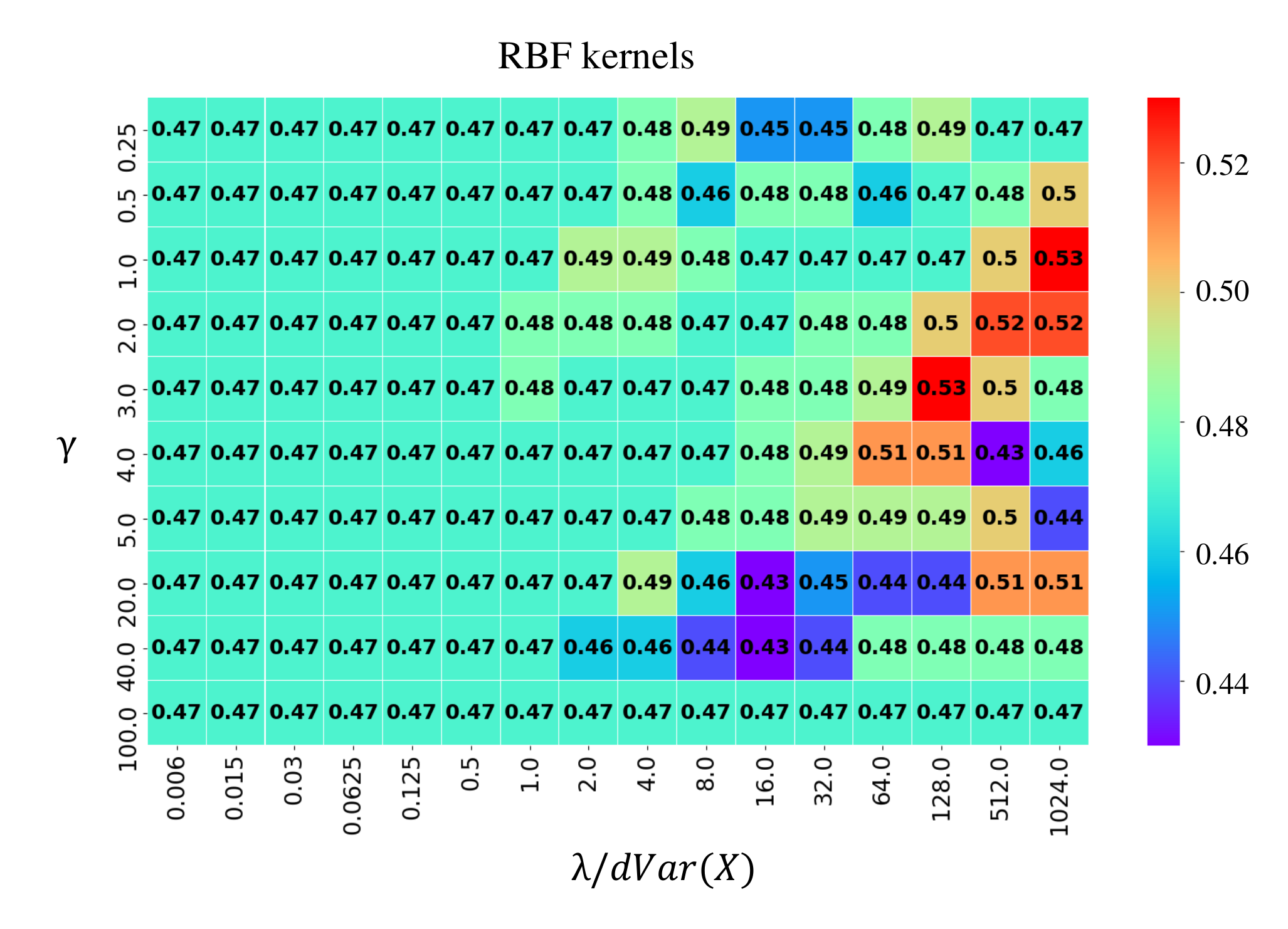}
\caption{\small{\textbf{ Performance of RBF kernels with varied $\lambda$ and $\gamma$.}}   The heat-map  presents the prediction accuracy (the higher means the better) of RBF kernels  with varied $\lambda$ and $\gamma$ given in Eqn.~\eqref{eq:lambda} and  Eqn.~\eqref{eq:gamma}. The feature dimension is $d=8$ and the sample size is $n=100$.} 
		\label{fig:fig3}
\end{figure}

Figure~\ref{fig:fig3} depicts the simulation results of RBF kernels with varied $\gamma$ and $\lambda$. Specifically, compared with the unregularized RBF kernels with fixed $\gamma$,  RBF kernels with tuned $\gamma$ and $\lambda$ achieve a better performance. For example, the prediction accuracy of RBF kernels is improved by $10\%$ (i.e., from $43\%$ to $53\%$) when the size of training data $n=100$ and the feature dimension $d=8$. These results reflect that the accuracy of RBF kernel can be improved by properly tuned $\gamma$ and $\lambda$.

	\subsection{Numerical simulation results}\label{append:subsec:sim-res}
In the main text, we present the core  results to indicate how the performance of quantum kernels is influenced by the imperfection of NISQ machines and how spectral transformation techniques can address this issue.  For completeness, here we illustrate more simulation  results to support our theoretical claims.

\begin{figure}[h!]
		\centering
		\includegraphics[width=0.98\textwidth]{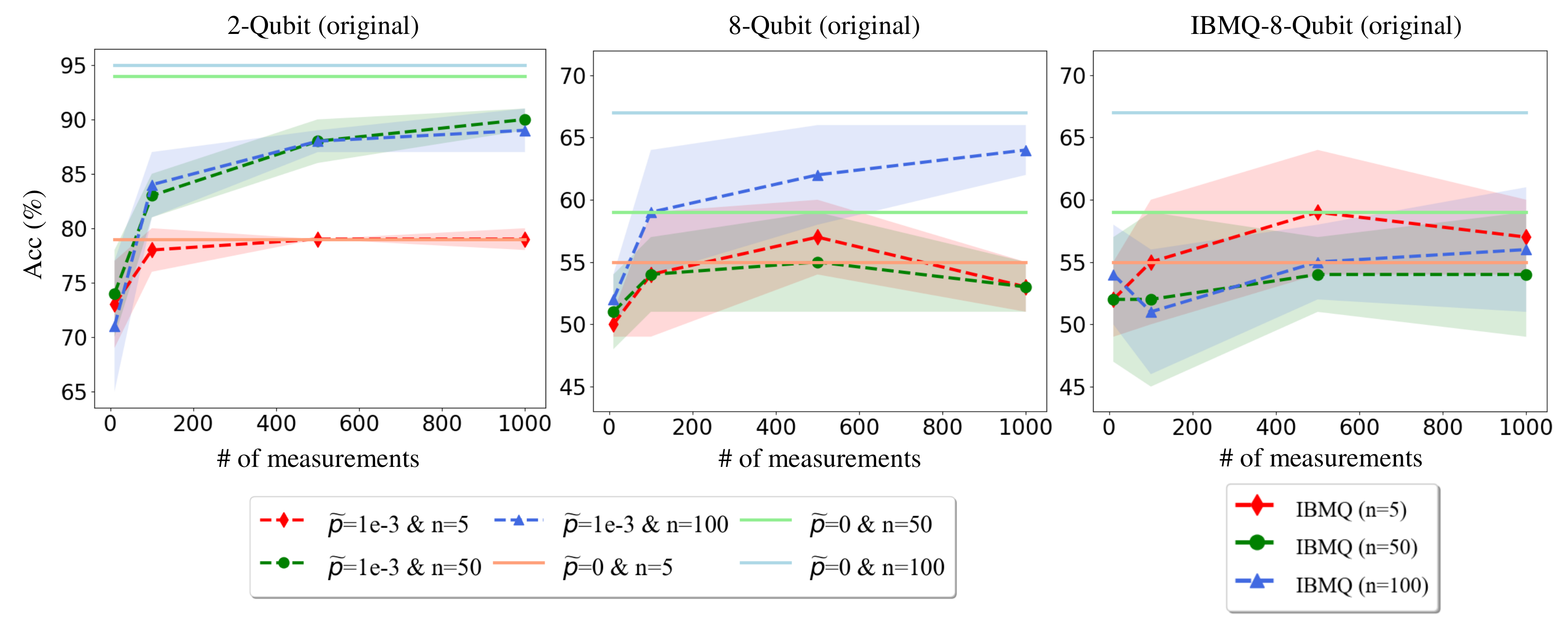}
		\caption{\small{\textbf{Performance of noisy quantum kernels with the varied size of the training set}. The left and middle panels present the prediction accuracy (the higher means the better) of noisy quantum kernels over the different size of the training set, i.e., $n \in \{5, 50, 100 \}$ for $N=2$ and $N=8$, respectively. The right panel shows the prediction accuracy of the quantum kernel carried out on the  IBMQ-Melbourne. The meaning of labels is the same with those explained in Figure~\ref{fig:2qubit}.}}
		\label{fig:Appfig1}
	\end{figure}
		
\textit{\underline{Performance of noisy quantum kernels.}}
We first benchmark the performance of quantum kernels under the depolarization noise. The hyper-parameters setting is as follows. The number of qubits is set as $N=2, 8$, the size of the training dataset is set as $n=5, 50, 100$, and the depolarization rate is set as $\widetilde{p}=0.001$. The simulation results are demonstrated in the left and middle panels in Figure~\ref{fig:Appfig1}, where the collected simulation results collaborate with Theorem \ref{thm:3.1}.  Statistically, the generalization error degrades with respect to the increased data size $n$ and the decreased number of measurements $m$. We note that when $N=8$ and $n=5, 50$, the increased number of measurements may not improve the performance of noisy quantum kernels. We suspect that this phenomenon is caused by 
the spectrum transformation losing some important information and the randomness of measurements. 

We next benchmark the performance of quantum kernels based on the noise model extracted from the real quantum-hardware, i.e., IBMQ-Melbourne.  The number of qubits is set as $N=8$ and the size of the training dataset is set as $n=5, 50, 100$. The simulation results are demonstrated in the right panel of Figure~\ref{fig:Appfig1}, which also accords with Theorem \ref{thm:3.1}. Namely, the generalization error becomes worse when the size of the training set is enlarged. 

\begin{figure}[h!]
	\centering
	\includegraphics[width=0.81\textwidth]{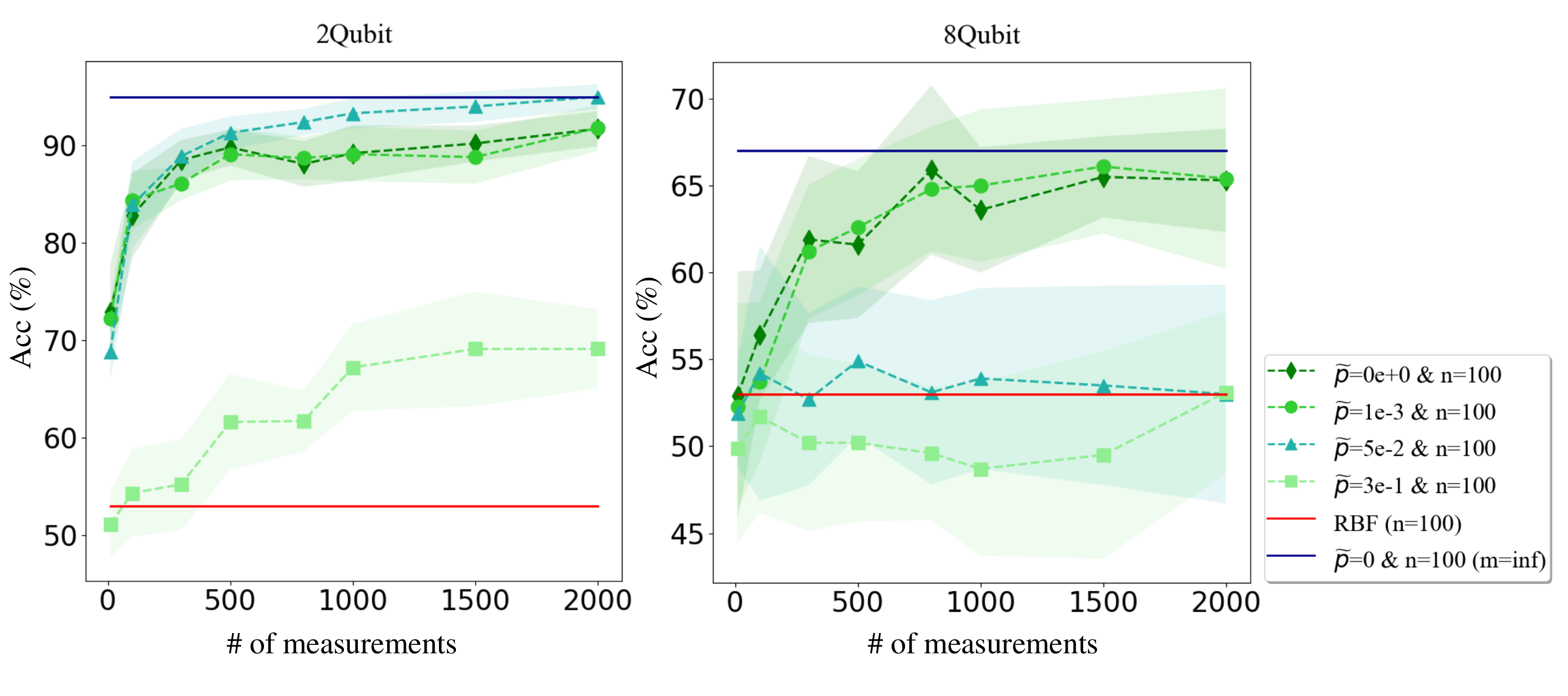}
		\caption{\small{\textbf{Performance of noisy quantum kernels for different depolarization rates for nearest projection.} The left and middle panels present the prediction accuracy (the higher means the better) of noisy quantum kernels over varied depolarizing rates $\widetilde{p}\in \{0, 0.001, 0.05, 0.3 \}$ and varied number of measurements $m=\{10, 100, 300, 500, 800, 1000, 1500, 2000\}$ when $N=2$ and $N=8$, respectively. The label `$m=\inf$' refers to an infinite number of measurements. The meaning of other labels is same with those explained in Figure~\ref{fig:2qubit}}. }
		\label{fig:Appfig2}
	\end{figure}

We last compare the performance between noisy quantum kernels and classical kernels, i.e., RBF, under different measurement shots and depolarization rates. The achieved results are exhibited in Figure~\ref{fig:Appfig2}. In particular, when the size of the training set is restricted  as $n = 100$, the noisy quantum kernel outperforms the classical RBF  when the depolarization rate $\widetilde{p}<0.05$ for both $N=2$ and $N=8$. These results provide a strong evidence to use quantum kernels to earn quantum advantages in the NISQ era.

		\begin{figure}[h!]
		\centering
	\includegraphics[width=0.98\textwidth]{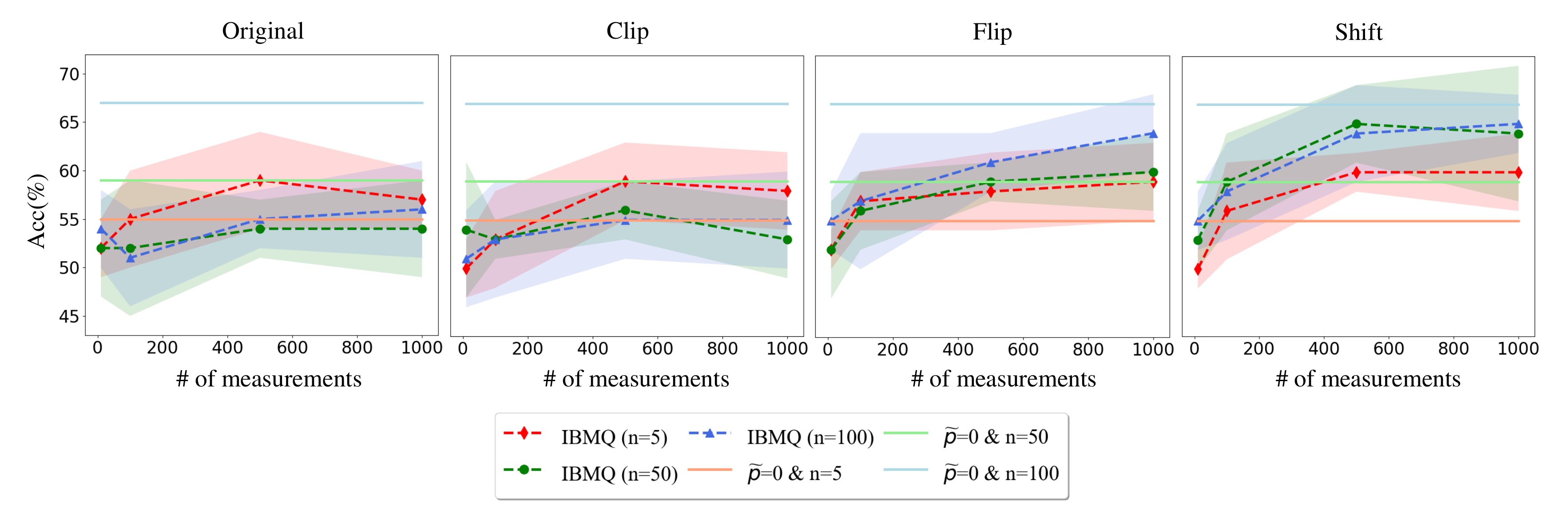}
		\caption{\small{\textbf{The comparison of quantum kernels under the IBMQ-Melbourne's noisy settings with different spectral transformation methods.} The  simulation results of noisy quantum kernels calibrated by the nearest projection as labeled by `Original', and the spectral transformation methods, i.e., clipping, flipping, and shifting methods, are shown in the 
		order from left to right respectively. The meaning of different labels is same with those explained in Figure~\ref{fig:2qubit}. }}
		\label{fig:Appfig3}
	\end{figure}

\textit{\underline{Performance of noisy quantum kernels with spectral transformation methods}.}
Before moving on to elaborate the appended simulation details, we first address the nearest projection technique employed in the main text. Using this technique facilitate to optimization since the indefinite kernels lead to a non-convex optimization problem. An intuitive approach is to compute a nearest positive definite matrix in terms of Frobenius norm by increasing the eigenvalues less than $\delta$ to $\delta$ with the threshold $\delta>0$.

 We now explore how the spectral transformation methods introduced in the main text improve the performance of quantum kernels based on the noisy model extracted from the real hardware, i.e., IBMQ-Melbourne. The hyper-parameters setting is as follows. The qubit count is set as $N=8$, the number of measurements is set as $m\in \{10, 100, 500, 1000\}$, and the size of the training set is $n=\{5, 50, 100\}$. The simulation results are shown in Figure \ref{fig:Appfig3}. For all settings, the shift method dramatically improve the performance of the noisy quantum kernels.  These results provide a strong evidence to explore advanced spectral transformation techniques to enhance the capabilities of quantum kernels in the NISQ era.

    \subsection{More simulation results with fine-grained measurements}
    
	\begin{figure}[h!]
	\centering
\includegraphics[width=0.98\textwidth]{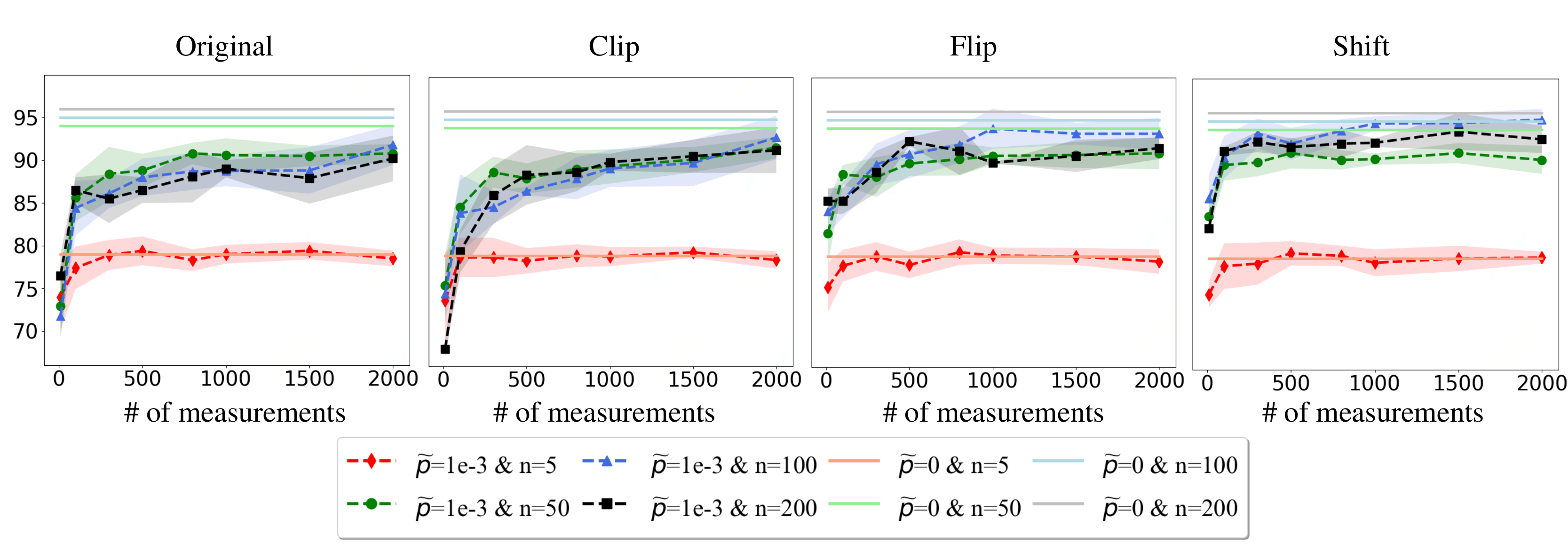}
	\caption{\small{\textbf{The comparison of noisy quantum kernels with different calibration methods.} }  The left subplot depicts the simulation results of noisy quantum kernels calibrated by the nearest projection as labeled by `Original'. The rest three subplots from left to right illustrates the simulation results calibrated by the spectral transformation methods, i.e., clipping, flipping, and shifting methods. The feature dimension is $d=2$. The meaning of different labels is the same with those explained in Figure~\ref{fig:2qubit}. }
	\label{fig:Appfig4}
\end{figure}  
    
To further support our theoretical claims, we conduct extensive numerical simulations by adding more points of measurement shots in the numerical simulation about  comparing noisy quantum kernels with different calibration methods. Specifically, there  are eight settings for the allowable number of measurements $m$, i.e., $\{10, 100, 300, 500, 800, 1000, 1500, 2000 \}$, while all other hyper-parameters settings are identical those introduced in the main text.

The simulation results are shown in Figure~\ref{fig:Appfig4}. Specifically, under such a fine-grained setting, the three calibration methods can still improve the prediction accuracy of noisy kernel compared with the original method. Moreover, the shifting method attains the best performance compared with clipping and flipping methods. These results echo with our theoretical results in the sense that suppressing the kernel approximation error can improve the generalization performance.

\end{document}